\documentclass[sigconf]{acmart}%,nonacm

\copyrightyear{2022}
\acmYear{2022}
\setcopyright{acmcopyright}
\acmConference[ASIA CCS '22] {Proceedings of the 2022 ACM Asia Conference on Computer and Communications Security}{May 30--June 3, 2022}{Nagasaki, Japan.}
\acmBooktitle{Proceedings of the 2022 ACM Asia Conference on Computer and Communications Security (ASIA CCS '22), May 30--June 3, 2022, Nagasaki, Japan}
\acmPrice{15.00}
\acmISBN{978-1-4503-9140-5/22/05}
\acmDOI{10.1145/3488932.3517390}
\AtBeginDocument{%
  \providecommand\BibTeX{{%
    \normalfont B\kern-0.5em{\scshape i\kern-0.25em b}\kern-0.8em\TeX}}}

\usepackage{amsmath}
\usepackage{amsthm}
\usepackage{graphicx}
\usepackage{booktabs} % To thicken table lines
\usepackage{enumitem}
\usepackage{footnote}
\usepackage{algorithm,tabularx}
\usepackage{algpseudocode}
\usepackage{placeins}
\usepackage{stfloats}
\newlength\myindent
\usepackage{numprint}
\usepackage{etoolbox}
\usepackage{hyperref}
\BeforeBeginEnvironment{appendices}{\clearpage}
\newtheorem{theorem}{Theorem}
\newtheorem{lemma}{Lemma}

\usepackage{tikz}
\usetikzlibrary{automata, positioning, arrows,calc,hobby,shapes,shapes.geometric,arrows}
\usetikzlibrary{shapes.multipart}
\usepackage{tkz-euclide}
\usepackage{subcaption}
\def\BibTeX{{\rm B\kern-.05em{\sc i\kern-.025em b}\kern-.08em T\kern-.1667em\lower.7ex\hbox{E}\kern-.125emX}}
\tikzstyle{int}=[draw, fill=lightgray, minimum height=1cm, minimum width=1.4cm,text width=1.3cm,align =center]
\theoremstyle{definition}

\begin{document}
\fancyhead{}
%%
%% The "title" command has an optional parameter,
%% allowing the author to define a "short title" to be used in page headers.
\title{Eliminating Sandwich Attacks with the Help of Game Theory}
%Avoiding Sandwich Attacks is Easy: The Game-Theory Behind Sandwich Attacks
%Traders can Extinguish Sandwich Attacks
%Sandwich Attackers Beware of Traders: They can Trample You
%Game-Theory Helps Traders Wipe Out Sandwich Attacks
%The Future of Sandwich Attacks Lies in the Hands of Traders
%An Analytical Approach to Eliminate Sandwich Attacks
%Eliminating Sandwich Attacks With Game Theory
%%
%% The "author" command and its associated commands are used to define
%% the authors and their affiliations.
%% Of note is the shared affiliation of the first two authors, and the
%% "authornote" and "authornotemark" commands
%% used to denote shared contribution to the research.

\author{Lioba Heimbach}
\affiliation{%
 \institution{ETH Zürich}
 \country{Switzerland}}
\email{hlioba@ethz.ch}
\author{Roger Wattenhofer}
\affiliation{%
 \institution{ETH Zürich}
 \country{Switzerland}}
\email{wattenhofer@ethz.ch}

%%
%% By default, the full list of authors will be used in the page
%% headers. Often, this list is too long, and will overlap
%% other information printed in the page headers. This command allows
%% the author to define a more concise list
%% of authors' names for this purpose.
%\renewcommand{\shortauthors}{Anonymous Author(s)}

%%
%% The abstract is a short summary of the work to be presented in the
%% article.
\begin{abstract}
    Predatory trading bots lurking in Ethereum's mempool present invisible taxation of traders on automated market makers (AMMs). AMM traders specify a slippage tolerance to indicate the maximum price movement they are willing to accept. This way, traders avoid automatic transaction failure in case of small price movements before their trade request executes. However, while a too-small slippage tolerance may lead to trade failures, a too-large slippage tolerance allows predatory trading bots to profit from sandwich attacks. These bots can extract the difference between the slippage tolerance and the actual price movement as profit.
    
    In this work, we introduce the \emph{sandwich game} to analyze sandwich attacks analytically from both the attacker and victim perspectives. Moreover, we provide a simple and highly effective algorithm that traders can use to set the slippage tolerance. We unveil that most broadcasted transactions can avoid sandwich attacks while simultaneously only experiencing a low risk of transaction failure. Thereby, we demonstrate that a constant auto-slippage cannot adjust to varying trade sizes and pool characteristics. Our algorithm outperforms the constant auto-slippage suggested by the biggest AMM, Uniswap, in all performed tests. Specifically, our algorithm repeatedly demonstrates a cost reduction exceeding a factor of 100.
\end{abstract}
%Traders specify their slippage tolerance, the maximum price movement they are willing to accept. If profitable, a trading bot executes a sandwich attack on the victim's transaction: by front- and back-running the trade. In the process, the bot forces the victim's trade to execute at an unfavorable price, which the slippage tolerance restricts. Hence, in general, the higher the slippage tolerance, the higher the victim's losses. However, a blockchain transaction does not execute upon submission but once included in a block. Slippage tolerance, thus, prevents transactions from automatically failing if the market moved before the block inclusion and simultaneously specifies the maximum acceptable market movement. Common belief states that higher slippage tolerances are necessary due to the crypto market's high volatility.
%%
%% The code below is generated by the tool at http://dl.acm.org/ccs.cfm.
%% Please copy and paste the code instead of the example below.
%%
\begin{CCSXML}
<ccs2012>
<concept>
<concept_id>10002978.10003006.10003013</concept_id>
<concept_desc>Security and privacy~Distributed systems security</concept_desc>
<concept_significance>500</concept_significance>
</concept>
<concept>
<concept_id>10003752.10010070.10010099.10010100</concept_id>
<concept_desc>Theory of computation~Algorithmic game theory</concept_desc>
<concept_significance>500</concept_significance>
</concept>
</ccs2012>
\end{CCSXML}
\ccsdesc[500]{Security and privacy~Distributed systems security}
\ccsdesc[500]{Theory of computation~Algorithmic game theory}

%%
%% Keywords. The author(s) should pick words that accurately describe
%% the work being presented. Separate the keywords with commas.
\keywords{blockchain, Ethereum, smart contract, decentralized finance, front-running, sandwich attack}

%% A "teaser" image appears between the author and affiliation
%% information and the body of the document, and typically spans the
%% page.

%%
%% This command processes the author and affiliation and title
%% information and builds the first part of the formatted document.
\maketitle

\section{Introduction}

In 2008, Nakatomo~\cite{nakamoto2008bitcoin} introduced Bitcoin, a fully decentralized currency based on cryptography. The introduction of smart contracts~\cite{wood2014ethereum} further fueled the initial excitement surrounding cryptocurrencies. Yet, apart from a few niche applications, cryptocurrencies mostly were alternative investment vehicles: you invest a dollar today and hope to have hundreds of dollars tomorrow. The emergence of \emph{decentralized finance (DeFi)} set off a new wave of interest for cryptocurrencies. DeFi is the first widespread application of cryptocurrencies and utilizes smart contracts running on a blockchain, currently mainly Ethereum~\cite{wood2014ethereum}, to offer financial services without relying on intermediaries. Traditional finance, on the other hand, relies on financial intermediaries, such as banks, brokerages, and exchanges. Thus, traditional finance requires users to trust intermediaries with their assets, while in DeFi, users have full autonomy over their assets. 

\emph{Decentralized exchanges (DEXes)} are a DeFi cornerstone. While centralized exchanges (CEXes) traditionally utilize the limit order book mechanism, matching individual sellers and buyers, traders do not need to be matched to a trading partner with opposite intentions in DEXes. Instead, trades on DEXes execute immediately upon inclusion in a block. When swapping two cryptocurrencies, the fluid exchange rate is determined algorithmically. Generally, the ratio and amount of the cryptocurrency pair, stored in the respective smart contract otherwise referred to as \emph{liquidity pool}, control the exchange rate. 

We observe users starting to acknowledge the benefits of DEXs. The 24hr trading volume of Uniswap~\cite{2021uniswap}, the most popular DEX, topped Coinbase's 24hr trading volume for the first time on 30 August 2020~\cite{uniswap_coinbase}, and repeatedly since. Further, all DEXes, which also include SushiSwap~\cite{2021sushiswap}, Curve~\cite{2021curve}, and dxdy~\cite{2021dxdy}, have more than \$27 billion locked as of 10 November 2021~\cite{defipulse}. 

Despite the undeniable rise in popularity of DeFi applications, the Ethereum peer-to-peer network is recently being characterized as a dark forest, with user transactions broadcast through the network being prey to predatory trading bots. The reason: the rise of DeFi on the Ethereum blockchain is testing some blockchain design principles. DeFi's smart contracts are dependent on transaction-ordering. One example of an attack that exploits the dependency on transaction-ordering are the omnipresent \emph{sandwich attacks} on DEX transactions. Sandwich attacks involve front- and back-running a  victim transaction -- presenting a tax on the victim's trade by forcing the trade to execute at an unfavorable price and then taking advantage of the created price difference. More than 84,000 transactions were sandwiched on Uniswap in April 2021 alone~\cite{zust2021analyzing}. Between May 2020 and May 2021, sandwich attacks earned at least 64,217 ETH~\cite{zust2021analyzing} -- presenting an invisible tax on trades.

A year ago, the risk presented to bots performing sandwich attacks was to time the front- and back-running transactions shortly before and after the victim's transaction. However, with the recent widespread adoption of flashbots~\cite{2021flash} by miners, sandwich attacks have become simpler than ever. Set out to light up the dark forest, \emph{front-running-as-a-service}, as offered by flashbots, allows attackers to execute sandwich attacks on victim transactions in the mempool virtually risk-free. Flashbots allows anyone submit victim transactions from the mempool directly to the miner along with the attack -- guaranteeing a successful attack.  

Thus far, the largest DEXes have not reacted to the risks sandwich attacks present to their users. The interfaces of Uniswap and SushiSwap both auto-suggest a fixed \emph{slippage} tolerance, the maximum acceptable price movement, ignorant to all trade parameters. In this paper, we show that a fixed slippage tolerance is unable to perform well, i.e., prevent sandwich attacks and avoid unnecessary transaction failures due to natural price fluctuations, consistently. Further, both platforms only warn their users of the risk of being front-run when inputting an exorbitant slippage tolerance. Even more startling, when users choose slippage tolerances that we find to be sensible, the two platforms issue warnings. Thus, their suggestions and (missed) warnings ramp up both the loss of their users and the profits of attackers.

\subsection{Our Contributions}
Our contribution is two-fold: 
\begin{enumerate}
    \item We analyze sandwich attacks by introducing the sandwich game. The sandwich game formalizes the sandwich attack problem from both the trader's and the bot's perspectives. 
    \item We provide AMM traders with an algorithm, allowing them to circumvent both unnecessary trade failures and most sandwich attacks. In the evaluation, we show that our algorithm outperforms the auto-slippage suggested by the biggest AMMs, in some cases by a factor of 100 and more.
\end{enumerate}

\section{Background}

DeFi now offers many centralized finance services. DeFi financial services are smart contracts on the blockchain -- the Ethereum blockchain hosts most services. DEXes are one momentous DeFi innovation to surface in recent years. However, DEXes present new challenges to blockchain design. While an order's position in a block used to be inconsequential for simple financial transactions, a transaction's relative position in a block is essential for many successful attacks. This section covers the preliminaries of sandwich attacks on DEXes.

\subsection{Ethereum Blockchain}

Ethereum is a public blockchain platform and the home to most DeFi applications, including DEXes such as Uniswap and SushiSwap. Users send their transactions to the \emph{mempool}: the waiting area for Ethereum transactions. Along with each transaction, users indicate the \emph{gas fee} (Ethereum's network transaction fee) they are willing to pay for their transaction. Transactions execute upon the inclusion in a block by a miner.

Until recently, users were over-bidding each other for block inclusion, and miners received the entire gas fee, but this changed with Ethereum's London Hard Fork update on 5 August 2021~\cite{2021eip1559}. As part of the London Hard Fork, EIP-1559 launched and changed Ethereum transactions fees. EIP-1559 aims to make transaction fees predictable, dividing the fee into the base fee and the priority fee. Automatically set by the network according to the current network load, the base fee is required for block inclusion and burned by the protocol. The priority fee, on the other hand, is collected by the miners. Thus, with EIP-1559, anyone wishing to make an Ethereum transaction will at least pay the base fee. 

\subsection{Automated Market Maker}
Most DEXes are automated market makers (AMMs). AMMs allow automatic trading of cryptocurrencies by an algorithm. Cryptocurrencies are aggregated in liquidity pools to facilitate this automated trading. A widespread adaptation of the AMM mechanism is Uniswap V2's, which we will introduce in the following. Note, however, that there are currently two active Uniswap versions: Uniswap V2 and Uniswap V3. This paper discusses and applies to both versions.

Uniswap allows the creation of liquidity pools between any cryptocurrency pair. Then, individual liquidity providers can deposit both cryptocurrencies at equal value in the respective pool. The liquidity aggregation allows traders to exchange the respective tokens in the pool. A transaction fee is levied for every trade and distributed pro-rata amongst the pool's liquidity providers. 

The AMM smart contract specifies the exchange rate offered to trades based on the number of tokens reserved in the liquidity pool. Uniswap utilizes a constant product market maker (CPMM), ensuring that product between the amounts of the two reserved pool currencies stays constant. We consider a liquidity pool between token $X$ and token $Y$, $X\rightleftharpoons Y$. The respective reserves are $x_t$ and $y_t$ at time $t$. A trader wishing to exchange $\delta _x$ tokens $X$ at time $t$ will receive $\delta _y$ tokens $Y$, where
\begin{equation}
        \delta  _{y}= y_{t} - \frac{x_{t} \cdot y_{t}}{x_{t}+(1-f)\delta_x} = \frac{y_{t} (1-f )\delta_x}{x_{t}+(1-f)\delta_x}, \label{eq:exchange}
\end{equation}
and $f$ is the transaction fee~\cite{adams2020uniswap}. The fee is charged on the input amount and is 0.3\% in the case of Uniswap.

Trades, however, are not executed immediately upon submission but are first sent to the mempool. Upon inclusion in a block, the trade executes. The delay between submission and execution implies that the pool reserves during execution are unknown to the trader when submitting the swap. Thus, traders indicate their \emph{slippage} tolerance -- the maximum acceptable price movement. 

We note that at the time of this writing, the interfaces of most major DEXes~\cite{UniswapWeb,2021sushiswap,2021curve,2021balancer,2021pancakeswap} suggest an auto-slippage, i.e., the same slippage tolerance is suggested for all transactions, independent of the size and the pool.

\subsection{Sandwich Attacks}
The aforementioned slippage tolerance simultaneously gives rise to sandwich attacks: front- and back-running victim transactions. Predatory traders listen to transactions in the public mempool and attack those that present profit opportunities by manipulating the transaction ordering and ensuring that one of the attack's transactions executes before the victim's transaction (front-running) and one after the victim's transaction (back-running).

To understand how sandwich attacks present a profit opportunity to bots, consider a victim's transaction trading 10 tokens $X$ for tokens $Y$ with 1\% slippage tolerance and 0.3\% transaction fee in a pool holding 100 tokens $X$ and 100 tokens $Y$. The trader is expected to receive 9.066 tokens $Y$ (Equation~\ref{eq:exchange}). However, after seeing the transaction in the mempool, a trading bot front-runs the victim by purchasing 0.524 tokens $Y$ with 0.529 tokens $X$. Thereby, the bot raises the price of token $Y$ for the victim to the limit indicated by the slippage tolerance. The following victim's trade subsequently only buys 8.975 tokens $Y$. Thus, the victim's trade is executed at a higher price than expected -- receiving exactly 1\% fewer tokens $Y$ than anticipated. The price of $Y$ is further increased by the victim's transaction. Finally, the bot concludes the attack by selling 0.524 tokens $Y$ at a higher price and receiving 0.635 tokens $X$. We observe that the bot's profit from the sandwich attack is 0.106 tokens $X$. Generally, the profitability of sandwich attacks increases with the victim's transaction size and slippage tolerance. 

Before the introduction of flashbots, bots were challenged to time their attacks through strategically setting the \emph{gas price}, the reward given to the miner for processing the transaction, such that a miner ordering the transactions according to gas price would sandwich the victim's transactions with the bots attack. As most miners used to sort transactions according to gas price, bots were able to predict the order of transactions. However, there was no guarantee. Further, this resulted in high gas prices for bots as they were overbidding each other in what is commonly referred to as \emph{priority gas auctions} (PGAs) -- the competitive bidding up of gas fees to obtain early block positions. 

Since the adaptation of flashbots, it is now possible for bots to guarantee that their attack transactions will sandwich the victim's transaction. Miners place bundles received from flashbots at the beginning of the block. Bots can submit the entire sandwich attack with the correct ordering to the miner by including the victim's transaction, detected in the mempool, in the bundle.

\section{Model}
There are three types of players in the sandwich game: traders, predatory trading bots, and miners. Traders send a DEX transaction to the mempool, representing potential bait to predatory trading bots. Predatory trading bots listen to these incoming transactions and launch a sandwich attack if they consider it profitable: aiming to front- and back-run the trader's transaction. Finally, miners select and order the transactions from the mempool in a block. From this point on, we will assume that the predatory trading bots are miners themselves or collude with miners -- allowing them to strategically order their transactions around the trader's transaction at no extra cost.

\subsection{Transaction Model}
We consider the liquidity pool between token $X$ and token $Y$, $X\rightleftharpoons Y$, with respective reserves $x_{0}$ and $y_{0}$ at time ${t_0}$. The current base fee for a Uniswap transaction is denoted by $b$. Note that while the base fee gives the minimum fee per gas, we utilize the minimum fee per Uniswap transaction in our analysis. This simplification is reasonable, as all individual Unisawp V2 transactions require approximately the same amount of gas. Thus, given the base fee per gas, we can compute the approximate base fee for a Uniswap transaction.

A trade $T_v$ to exchange $\delta _x$ tokens $X$ entering the mempool at time ${t_0}$ is identified as $T_v=(\delta_{v_x},s,f,b,x_{0},y_{0},{t_0})$. Here, $s$ is the specified slippage tolerance and $f$ the transaction fee. The trade would output $\delta _{v_y}$ tokens $Y$ if executed at time ${t_0}$, where 
\begin{equation*}
    \delta  _{v_y}= y_{0} - \frac{x_{0} \cdot y_{0}}{x_{0}+(1-f)\delta_x} = \frac{y_{0} (1-f )\delta_x}{x_{0}+(1-f)\delta_x}. 
\end{equation*}
Time advances when a trade is executed in pool $X\rightleftharpoons Y$, and as the transaction might not execute at time ${t_0}$, $\delta  _{v_y}$ is only an estimate. Assuming that the trade executed at time ${t_1}$, the trader will receive 
\begin{equation*}
    \tilde{\delta}  _{v_y} = \frac{y_{1} (1-f )\delta_x}{x_{1}+(1-f)\delta_x}, 
\end{equation*}
tokens $Y$. Depending on the changes in the pool reserves between time ${t_0}$ and ${t_1}$, the trader might receive more or less tokens $Y$. In order to control how bad the exchange rate becomes for the traders, they specify a slippage tolerance $s$. The trade will only execute at time $t_1$, if
$$ \tilde{\delta} _{v_y} \geq (1-s) \delta_{v_y}.$$
Otherwise, the trade will fail to execute. With the sandwich game, we analyze how traders optimally set the slippage tolerance to achieve a low expected trade cost.

\subsection{Attack Model}

The predatory trading bot listens to the inflowing transactions in the mempool. Upon noticing the trade $T_v=(\delta_{v_x},s,f,b,x_{0},y_{0},{t_0})$ entering the mempool, the predatory trading bot computes the optimal input for the sandwich attack ($\delta_{a_{x}}$) and assess whether the attack will be profitable. We assume optimal conditions for the predatory trading bot: access to unlimited funds, guaranteed transaction ordering, and only paying the base fee. Assuming that the predatory trading bot has access to unlimited funds is reasonable and represents the worst case for traders. Additionally, letting the miner be the predatory trading bot again represents the worst case for traders. Further, it allows the trading bot to only pay the base fee for its transactions and control transaction ordering. We further assume that the trading bot takes the front-running transaction's output as the input of the back-running transaction. 
%Max profit, input of trade

\section{Sandwich Game}\label{sec:game}
\begin{figure}[t]
    \begin{subfigure}[]{0.21\textwidth}
        \centering
        \begin{tikzpicture}[auto,>=latex',every text node part/.style={align=center}]
        \node[cylinder,draw=white,thick,aspect=0.4,minimum height=1.4cm,minimum width=1.3cm,shape border rotate=90,cylinder uses custom fill, cylinder body fill=white!30,cylinder end fill=white!10,text=white] at (0,1.6) (A) {$x_{0} X$\\ $y_{0} Y$ };
            \node[cylinder,draw=black,thick,aspect=0.4,minimum height=1.4cm,minimum width=1.3cm,shape border rotate=90,cylinder uses custom fill, cylinder body fill=teal!30,cylinder end fill=teal!10] at (0,0) (A) {$x_{0} X$\\ $y_{0} Y$ } ;
            \path[->,text width=1 cm, align = center] (-1.6,0) edge node {$\delta_{v_x}$} (-0.7,0);
            \path[->,text width=1 cm, align = center] (0.7,0) edge node {$\delta_{v_y}$} (1.6,0);
           %\node[cylinder,draw=black,thick,aspect=0.4,minimum height=1.4cm,minimum width=1.3cm,shape border rotate=90,cylinder uses custom fill, cylinder body fill=purple!30,cylinder end fill=purple!10] at (0,-1.6) (A) {$x_{2} X$\\ $y_{2} Y$ } ;
           \node[cylinder,draw=white,thick,aspect=0.4,minimum height=1.4cm,minimum width=1.3cm,shape border rotate=90,cylinder uses custom fill, cylinder body fill=white!30,cylinder end fill=white!10,text=white] at (0,-1.6) (A) {$x_{2} X$\\ $y_{2} Y$ } ;
        \end{tikzpicture}
        \caption{Illustration of an ordinary Uniswap transaction.}\label{fig:nosand}
    \end{subfigure}
    \hfill
    \begin{subfigure}[]{0.23\textwidth}
        \centering
        \begin{tikzpicture}[auto,>=latex',every text node part/.style={align=center}]
        \node[cylinder,draw=black,thick,aspect=0.4,minimum height=1.4cm,minimum width=1.3cm,shape border rotate=90,cylinder uses custom fill, cylinder body fill=purple!30,cylinder end fill=purple!10] at (0,1.6) (A) {$x_{0} X$\\ $y_{0} Y$ } ;
            \path[->,text width=1 cm, align = center] (-1.6,1.6) edge node {$\delta^{\text{in}}_{a_{x}}$} (-0.7,1.6);
            \path[->,text width=1 cm, align = center] (0.7,1.6) edge node {$\delta_{a_{y}}$} (1.6,1.6);
            \node[cylinder,draw=black,thick,aspect=0.4,minimum height=1.4cm,minimum width=1.3cm,shape border rotate=90,cylinder uses custom fill, cylinder body fill=teal!30,cylinder end fill=teal!10] at (0,0) (A) {$x_{1} X$\\ $y_{1} Y$ } ;
            \path[->,text width=1 cm, align = center] (-1.6,0) edge node {$\delta_{v_x}$} (-0.7,0);
            \path[->,text width=1 cm, align = center] (0.7,0) edge node {$\tilde{\delta}_{v_y}$} (1.6,0);
            \node[cylinder,draw=black,thick,aspect=0.4,minimum height=1.4cm,minimum width=1.3cm,shape border rotate=90,cylinder uses custom fill, cylinder body fill=purple!30,cylinder end fill=purple!10] at (0,-1.6) (A) {$x_{2} X$\\ $y_{2} Y$ } ;
            \path[->,text width=1 cm, align = center] (-1.6,-1.6) edge node {$\delta_{a_{y}}$} (-0.7,-1.6);
            \path[->,text width=1 cm, align = center] (0.7,-1.6) edge node {$\delta^{\text{out}}_{a_{x}}$} (1.6,-1.6);      
        \end{tikzpicture}
        \caption{Illustration of a sandwich attacked Uniswap transaction.}\label{fig:sand}
    \end{subfigure}\vspace{-8pt}
\caption{Illustration of a sandwich attack. In Figure~\ref{fig:nosand} the trade executes without being attacked, while the trade is front- and back-run in Figure~\ref{fig:sand}.}
\label{architecture}
\vspace{-0.3cm}
\end{figure}
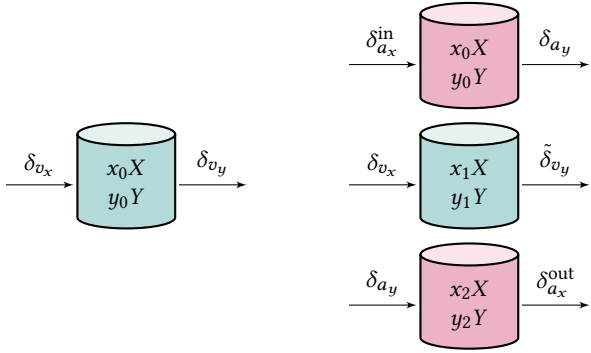

We start by going through the general mechanism of the game. The victim submits a transaction $T_v$ wishing to exchange $\delta_{v_x}>0$ in pool $X\rightleftharpoons Y$, with respective reserves $x_{0}>0$ and $y_{0}>0$. The pools transaction fee is $f$ ($0\leq f<1$) and the transaction's slippage tolerance is $s$ ($0<s<1$). Transaction $T_v=(\delta_{v_x},s,f,b,x_{0},y_{0},{t_0})$ enters the mempool at time ${t_0}$. When submitting the trade, the victim expects $\delta_{v_y}$ tokens $Y$. $\delta_{v_y}$ corresponds to the number of tokens the victim would receive if no other trade is executed beforehand, i.e., if the reserves in the pool do not shift in the meantime (cf. Figure~\ref{fig:nosand}). Thus,
$$\delta_{v_y} = \frac{y_{0}(1-f )\delta_{v_x}}{x_{0}+(1-f)\delta_{v_x}}.$$
However, when a sandwich attack occurs, the predatory bot first executes a transaction $T_{A_1}$ exchanging $\delta^{\text{in}}_{a_{x}}>0$ tokens $X$ for $\delta _{a_y}$ tokens $Y$, where $$\delta _{a_y} = \frac{y_{0}(1-f )\delta^{\text{in}}_{a_{x}}}{x_{0}+(1-f)\delta^{\text{in}}_{a_{x}}}.$$
Now the victims transaction executes at time $t_1$, assuming that the slippage tolerance is not overshot, with unfavourable reserves $x_{1}= x_{0}+\delta^{\text{in}}_{a_{x}}$ and $y_{1}= (x_{0}y_{0})/(x_{0}+\delta^{\text{in}}_{a_{x}}(1-f))$. The victim only receives
$$\tilde{\delta} _{v_y} = \frac{y_{1}(1-f )\delta_{v_x}}{x_{1}+(1-f)\delta_{v_x}}= \frac{\frac{x_{0} y_{0}}{x_{0} \delta^{\text{in}}_{a_{x}}} (1-f )\delta_{v_x}}{x_{0}+\delta^{\text{in}}_{a_{x}}+(1-f)\delta_{v_x}}$$
tokens $Y$, and $\tilde{\delta} _{v_y}<\delta _{v_y}$. To finish the attack, the bot exchanges $\delta _{a_y}$ tokens $Y$ at time $t_2$ with pool reserves $x_{2}= x_{1}+\delta_{v_x}$ and $y_{2}= (x_{1}y_{1})/(x_{1}+\delta_{v_{x}}(1-f))%=\frac{(x_{0}+\delta^{\text{in}}_{a_{x}})\left(\frac{x_{0}y_{0}}{x_{0}+\delta^{\text{in}}_{a_{x}}(1-f)}\right)}{(x_{0}+\delta^{\text{in}}_{a_{x}})+\delta_{v_{x}}(1-f)}
$. In this final transaction $T_{A_2}$, the bot receives
$$\delta^{\text{out}}_{a_{x}} = \frac{x_{2}(1-f )\delta_{a_y}}{y_{2}+(1-f)\delta_{a_y}}= \frac{( x_{0}+\delta^{\text{in}}_{a_{x}}+\delta_{v_x})(1-f )\delta_{a_y}}{\frac{(x_{0}+\delta^{\text{in}}_{a_{x}})\left(\frac{x_{0}y_{0}}{x_{0}+\delta^{\text{in}}_{a_{x}}(1-f)}\right)}{(x_{0}+\delta^{\text{in}}_{a_{x}})+\delta_{v_{x}}(1-f)}+(1-f)\delta_{a_y}}$$
tokens $X$. The bots profit $P_a$ is
$$P_a=\delta^{\text{out}}_{a_{x}}-\delta^{\text{in}}_{a_{x}} -2b,$$ where $b$ is the base fee in the currency $X$. Note, that the base fee $b\geq0$ is fixed for a block and known.

\subsection{Adversary Perspective}

We start by finding the optimal strategy for a predatory trading bot and first find the best attack input $\delta^{\text{in}}_{a_{x}}$: the input maximizing  $P_a=\delta^{\text{out}}_{a_{x}}-\delta^{\text{in}}_{a_{x}}-2b$. For this, we consider an arbitrary victim transaction $T_v=(\delta_{v_x},s,f,b,x_{0},y_{0},{t_0})$ in pool $X\rightleftharpoons Y$. First, we consider the case, where $s$ is not set for the transaction, i.e. $s=1$. We show that bot can then analytically determine the the optimal sandwich attack size in Lemma~\ref{lem:slipda}.

\begin{lemma}\label{lem:slipda}
    We can analytically determine the trading bot's optimal input, ($\delta^o_{a_x}$) when the victim's transaction $T_v$ indicates no slippage tolerance, i.e., $s=1$.
\end{lemma}
\begin{proof} 
    To maximize $P_a$, it is sufficient for the bot to maximize $\delta^{\text{diff}}_a=\delta^{\text{out}}_{a_{x}}-\delta^{\text{in}}_{a_{x}}$. We start by finding the zero crossing of the derivative of $\delta^{\text{diff}}_a$ with respect to $\delta^{\text{in}}_{a_{x}}$.

    \begin{align*}
        &\frac{\partial \delta^{\text{diff}}_a}{\partial \delta^{\text{in}}_{a_{x}}}\\
        &=\frac{x_{0}({\delta^{\text{in}}_{a_{x}}}^2 f(\delta_{v_x} (1-f)^2 -(2-f)x_{0})}{\left(\delta^{\text{in}}_{a_{x}} (1 - f)^2 (\delta^{\text{in}}_{a_{x}} + \delta_{v_x} - 
      \delta_{v_x}  f) + \delta^{\text{in}}_{a_{x}} (2 - (2 - f) f) x_{0} + x_{0}^2 \right)^2}\\
      &+\frac{ 2 \delta^{\text{in}}_{a_{x}} x_{0} (\delta_{v_x} (1 - f)^2 -(2 -f) f\cdot  x_{0}) )}{\left(\delta^{\text{in}}_{a_{x}} (1 - f)^2 (\delta^{\text{in}}_{a_{x}} + \delta_{v_x} - 
      \delta_{v_x}  f) + \delta^{\text{in}}_{a_{x}} (2 - (2 - f) f) x_{0} + x_{0}^2 \right)^2}\\
      &+ \frac{x_{0}^2 (\delta_{v_x}^2 (1 -f)^3 + \delta_{v_x} (2-f)(1-f)^2 x_{0} - (2 - 
           f) f x_{0}^2)}{\left(\delta^{\text{in}}_{a_{x}} (1 - f)^2 (\delta^{\text{in}}_{a_{x}} + \delta_{v_x} - 
      \delta_{v_x}  f) + \delta^{\text{in}}_{a_{x}} (2 - (2 - f) f) x_{0} + x_{0}^2 \right)^2} 
    \end{align*}
    
    The single zero crossing of $\partial \delta^{\text{diff}}_a/\partial \delta^{\text{in}}_{a_{x}}$, such that $ \delta^{\text{in}}_{a_{x}} >0$ is located at
    \begin{align*}
       \delta^{\text{o}}_{a_{x}} =& \frac{(
     \delta_{v_x} (1 - f)^2 x_{0} - (2 - f) f x_{0}^2 }{(
     (2 -f) f x_{0}- \delta_{v_x} (1 - f)^2 f  ) } \\&+\frac{ \sqrt{
      \delta_{v_x}^2 (1 - f)^3 x_{0} (x_{0} - (1 - f)^2 f (\delta_{v_x} + x_{0}))})}{(
     (2 -f) f x_{0}- \delta_{v_x} (1 - f)^2 f  ) } .
    \end{align*}
    $$$$
     
   As $$\left. \frac{\partial ^2\delta^{\text{diff}}_a}{ {\partial \delta^{\text{in}}_{a_{x}}} ^2}\right\vert_{\delta^{\text{o}}_{a_{x}}}<0,$$
   $\delta^o_{a_x}$ is the trading bot's optimal input.
\end{proof}

While one might expect that the input of the optimal attack $\delta_{a_x}$ to be infinite, Lemma~\ref{lem:slipda} shows that this is not the case. The optimal input amount is limited, as the bot performing the sandwich attack needs to pay the fee $f$ twice, which increases with the input amount.

However, in most cases, traders will specify the slippage tolerance $s$, further limiting the maximum input size of the bots attack. In Lemma~\ref{lem:slips} we find the bot's maximal input such that the trade executes and show that the bot can compute it analytically. 

\begin{lemma}\label{lem:slips}
    The bot's maximal input ($\delta^s_{a_x}$) for a transaction exchanging $\delta _x$ tokens $X$ with slippage tolerance $s$ such that the victim's trade still executes can be calculated analytically and is given in the proof.
\end{lemma}
\begin{proof} 
We consider a sandwich attack with initial input $\delta_{a_{x}}$, changing the pool reserves from $x_{0}$ to $x_{1}=x_{0}+\delta_{a_{x}}$ tokens $X$ and from $y_{0}$ to $y_{1}=(x_{0} y_{0})/(x_{0} +\delta_{a_{x}})$ tokens $Y$. The new output of the victim transaction, assuming that it goes through will be
\begin{align*}
    \tilde{\delta} _{v_y} &= \frac{y_{1}(1-f )\delta_{v_x}}{x_{1}+(1-f)\delta_{v_x}}= \frac{\frac{x_{0} y_{0}}{x_{0} +\delta_{a_{x}}} (1-f )\delta_{v_x}}{x_{0}+\delta_{a_{x}}+(1-f)\delta_{v_x}},
\end{align*}     
and the victims transaction will go through, if
\begin{align*}
    \tilde{\delta} _{v_y} &\geq (1-s) \delta_{v_y}\\
    \frac{\frac{x_0 y_0}{x0 +\delta_{a_{x}}} (1-f )\delta_{v_x}}{x_0+\delta_{a_{x}}+(1-f)\delta_{v_x}} &\geq (1-s) \frac{y_0 (1-f )\delta_{v_x}}{x_0+(1-f)\delta_{v_x}}.
\end{align*}
Thus, the bot's maximal input ($\delta_{a_x}^s$) increases the slippage incurred by the victim to its tolerance, i.e., $\tilde{\delta} _{v_y} = (1-s) \delta_{v_y}$. Solving for $\delta_{a_x}^s$, we find that the maximal input is 
$$ \delta_{a_x}^s  = \frac{\frac{\sqrt{n(x_{0},f,\delta_{v_x},s)}}{1-s}-\delta_{v_x} (1-f)^3 -(2-f)(1-f)x_{0}  }{2(1-f)^2},$$
where
\begin{align*}
        n(x_{0},f,\delta_{v_x},s)=&(1 - f)^2 (1 - s) (\delta_{v_x}^2 (1-f)^4 (1 - s) \\
        &+ 2 \delta_{v_x} (1 -f)^2 (2 - f (1 - s))x_{0}  \\
        &+ (4 - f (4 - f (1 - s))) x_{0}^2.\qedhere
\end{align*}
\end{proof}

Following from Lemma~\ref{lem:slipda} and Lemma~\ref{lem:slips}, we find that the bot's optimal input is $\delta^{\text{in}}_{a_{x}} = \min \{\delta_{a_x}^o,\delta_{a_x}^s\}$ in Theorem~\ref{thm:da}. In case the profit of the corresponding attack is negative, not profitable attack exitst and the bot does not execute any attack.
\begin{theorem}\label{thm:da}
    The bot's optimal input is 
    $\delta^{\text{in}}_{a_{x}} = \min \{\delta_{a_x}^o,\delta_{a_x}^s\}.$
\end{theorem}
\begin{proof} 
    From Lemma~\ref{lem:slipda} we know, that the there is a single maximum $\delta_{a_x}^o$, such that $\delta^{\text{in}}_{a_{x}}>0$. However, in case the victim's trade does not execute for $\delta^{\text{in}}_{a_{x}} = \delta_{a_x}^o$, the maximum will be at the endpoint of the permitted interval ($\delta_{a_x}^s$). Thus, $\delta^{\text{in}}_{a_{x}} = \min \{\delta_{a_x}^o,\delta_{a_x}^s\}$.
\end{proof}

In Figure~\ref{fig:profit}, we show the effects of the slippage tolerance, transaction fee, and trade size in relation to the pool size on a bots profit. Note that the simulation in Figure~\ref{fig:profit} disregards the base fee, which would remove the constant amount ($2b$) from the profit. Figure~\ref{fig:sizeandslip} demonstrates that, as expected, a bot's maximum profit is dependent on both the slippage tolerance and transaction size. For small transaction sizes, even transactions with high slippage tolerances are not attackable. Additionally, higher transaction fees allow higher slippage tolerances before the trades become attackable. Thus, the constant auto-slippage, independent of transaction size and transaction fee, suggested by Uniswap and SushiSwap, appears counter-intuitive.

\begin{figure}[t]
\centering
    \begin{subfigure}{\linewidth}
        \includegraphics[width = \textwidth]{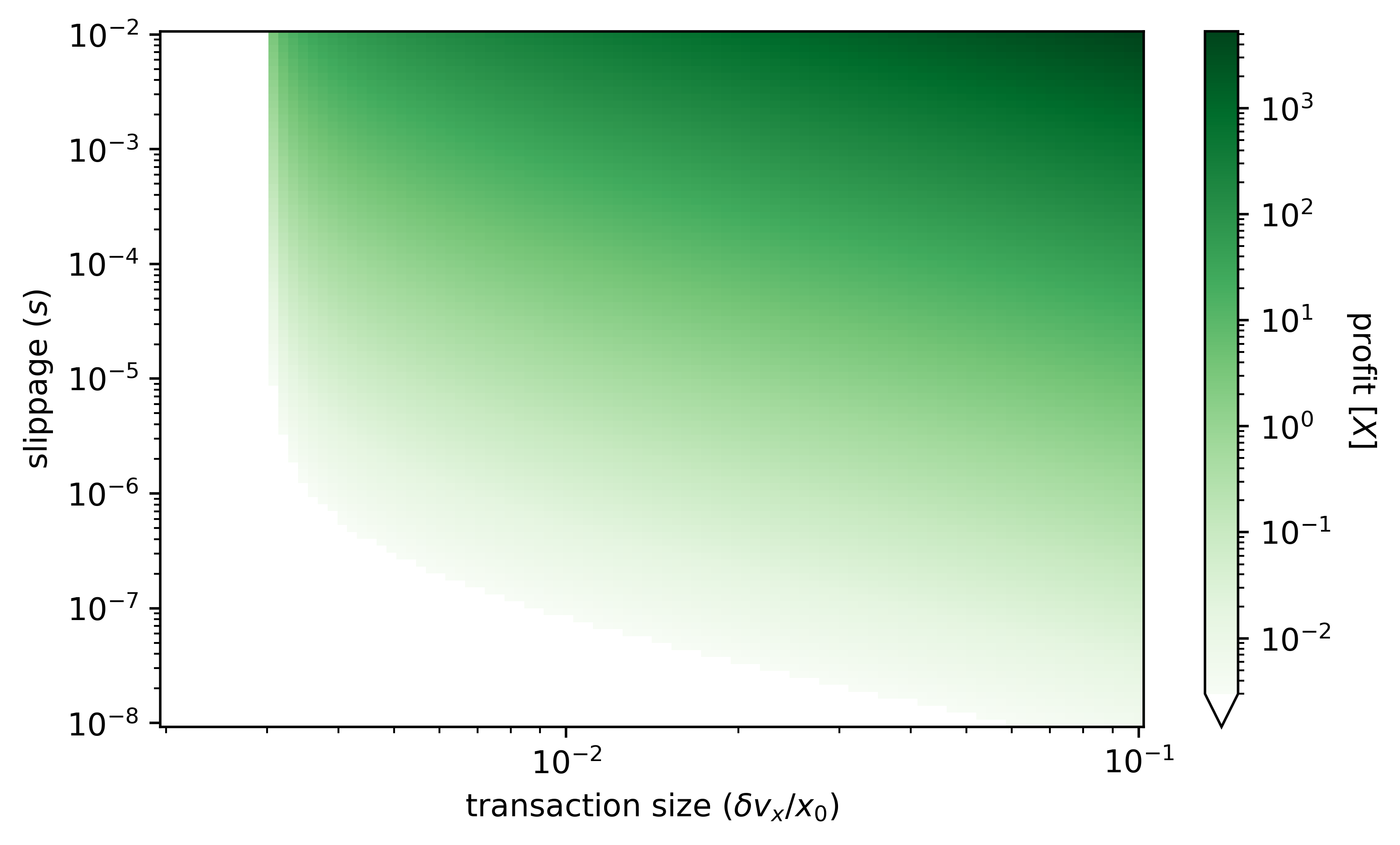}\vspace{-6pt}
        \caption{Effects of slippage tolerance ($s$) and transaction size in relation to pool size ($\delta_{v_x}/x_{0}$) on the bot's maximal profit. We set $f=0.003$.} \label{fig:sizeandslip}
    \end{subfigure}%

    \begin{subfigure}{\linewidth}
        \includegraphics[width = \textwidth]{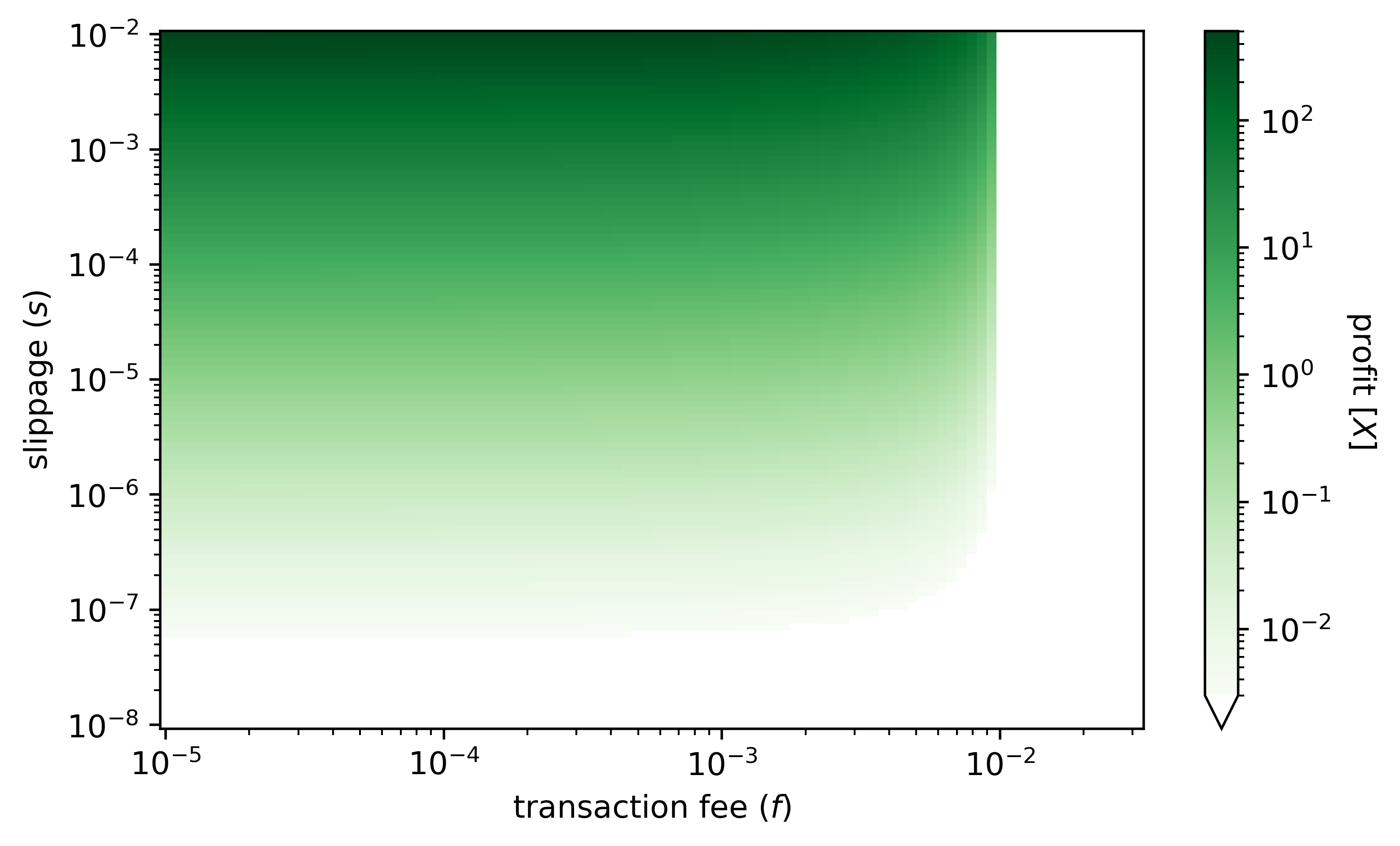}\vspace{-6pt}
        \caption{Effects of slippage tolerance ($s$) and transaction fee ($f$) on the bot's maximal profit. We set $\delta_{v_x}/x_{0}=0.01$.} \label{fig:feeandslip}
    \end{subfigure}\vspace{-10pt}
\caption{The effects of slippage tolerance ($s$), transaction fee ($f$), and transaction size in relation to pool size ($\delta_{v_x}/x_{0}$) on a predatory trading bot's maximal profit for a victim's trade $T_v=(\delta_{v_x},s,f,b,x_{0},y_{0},{t_0})$. We disregard the base fee, set $x_{0}$ to 5000000 $X$ and give the profit in the currency $X$.}\label{fig:profit}
\vspace{-0.2cm}
\end{figure}

Further, we show in Theorem~\ref{thm:profit} that the bot's maximum profit cannot exceed the victim's loss. We will rely on this result in Section~\ref{sec:trader} to allow for straightforward computations on the trader side. Seeing that the bot's maximum profit can trail the victim's loss, we wonder where the remaining profit is collected. By noticing that sandwich attacks increase the volume in the pool, we conclude that liquidity providers also profit through sandwich attacks. 
\begin{theorem}\label{thm:profit}
    The bot's profit cannot exceed the victim's loss.
\end{theorem}
\begin{proof} 
    Without loss of generality, we assume both the transaction fee $f$ and the base fee $b$ to be zero. Both would only decrease the bot's profit. Note, that without transaction fees, the slippage tolerance will restrict the bot's optimal input for all $s\neq 1$. We will start by analyzing the case where the victim  trade $T_v=(\delta_{v_x},s,f,b,x_{0},y_{0},{t_0})$ sets the slippage tolerance $s\neq 1$, for $f$ the bot's optimal input (Lemma~\ref{lem:slips}) becomes
    %$$ \delta^{\text{in}}_{a_{x}}  = \frac{1}{2}\left( \frac{\sqrt{(1-s) ( \delta_{v_x}^2 (1-s) +4  \delta_{v_x}x_{0}+4x_{0}^2) }}{1-s}- 2x_{0} -\delta_{v_x}\right).$$
$$
\delta^{\text{in}}_{a_{x}}  = \frac{1}{2}
\left(\vcenter{\hbox{$\displaystyle
\frac{\sqrt{(1-s) ( \delta_{v_x}^2 (1-s) +4  \delta_{v_x}x_{0}+4x_{0}^2) }}{1-s}- 2x_{0} -\delta_{v_x}
$}}\right).
$$
    The bot's profit is then given by

    \begin{align*}
        P_a =& \delta^{\text{out}}_{a_{x}} -\delta^{\text{in}}_{a_{x}}= \frac{ \delta_{v_x} s \left( \delta_{v_x}+x_{0}\right)}{\delta_{v_x} s+x_{0}},
        %=& \frac{( x_{0}+\delta^{\text{in}}_{a_{x}}+\delta_{v_x})\delta_{a_y}}{\frac{\left(x_{0}+\delta^{\text{in}}_{a_{x}}\right)\left(\frac{x_{0}y_{0}}{x_{0}+\delta^{\text{in}}_{a_{x}}}\right)}{\left(x_{0}+\delta^{\text{in}}_{a_{x}}\right)+\delta_{v_{x}}}+\delta_{a_y}}-\delta^{\text{in}}_{a_{x}}\\
    \end{align*}
    while the victims loss is given as
    \begin{align*}
        L_v &= s \cdot \delta_{v_y} \frac{x_{2}}{y_{2}}
        %&=   \delta_{v_y} s\frac{\left(\delta_{v_x}(1-s) + \sqrt{(1-s)\left (\delta_{v_x} ^2 (1-s) +4  \delta_{v_x}x_{0}+4x_{0}^2 \right)} \right)^2}{4(1-s)^2 x_{0}( \delta_{v_x}+x_{0})}\\
        =  s \cdot \delta_{v_x} \frac{\left(\delta_{v_x} + \sqrt{\left (\delta_{v_x} ^2  +\frac{4  \delta_{v_x}x_{0}+4x_{0}^2}{1-s}  \right)} \right)^2}{4 x_{0}( \delta_{v_x}+x_{0})} .
    \end{align*}
   
    To obtain the loss $L_v$, we multiply the victim's loss in tokens $Y$ ($s \cdot \delta_{v_y}$) by the price $p_{y\rightarrow x}$ at the time the victim's losses were realized, such that it is in the same currency as the bot's profit.  
    To show that the profit cannot exceed the loss, we show $P_a/L_v\leq 1$.
    \begin{align*}
        P_a/L_v 
        &= \frac{4x_{0}( \delta_{v_x}+x_{0})^2}{(\delta_{v_x}s +x_{0})\left(\delta_{v_x} + \sqrt{\left (\delta_{v_x} ^2 +\frac{4  \delta_{v_x}x_{0}+4x_{0}^2}{1-s} \right)} \right)^2}\\
        &\leq \frac{4x_{0}( \delta_{v_x}+x_{0})^2}{(\delta_{v_x}s +x_{0})\left(\delta_{v_x} + (\delta_{v_x}+2 x_{0})\right)^2} =\frac{x_{0}}{\delta_{v_x} s+ x_{0}}
        \leq 1.
    \end{align*}
    We turn the the case where $s =1$. The bot's optimal input is then $\delta^{\text{in}}_{a_{x}} \rightarrow \infty$ and we find the associated profit to be
    \begin{align*}
        \lim_{\delta^{\text{in}}_{a_{x}} \rightarrow \infty}  P_a&=\lim_{\delta^{\text{in}}_{a_{x}} \rightarrow \infty}  ( \delta^{\text{out}}_{a_{x}} -\delta^{\text{in}}_{a_{x}}) \\
        &= \lim_{\delta^{\text{in}}_{a_{x}} \rightarrow \infty} \frac{\delta_{v_x} \delta^{\text{in}}_{a_{x}} (\delta_{v_x}+ 2x_{0}+\delta^{\text{in}}_{a_{x}})}{\delta_{v_x} \delta^{\text{in}}_{a_{x}}  +(\delta^{\text{in}}_{a_{x}}+x_{0})^2 }+x_{0})^2 = \delta_{v_x}.
        %=& \frac{( x_{0}+\delta^{\text{in}}_{a_{x}}+\delta_{v_x})\delta_{a_y}}{\frac{\left(x_{0}+\delta^{\text{in}}_{a_{x}}\right)\left(\frac{x_{0}y_{0}}{x_{0}+\delta^{\text{in}}_{a_{x}}}\right)}{\left(x_{0}+\delta^{\text{in}}_{a_{x}}\right)+\delta_{v_{x}}}+\delta_{a_y}}-\delta^{\text{in}}_{a_{x}}\\
    \end{align*}
    Further, for  $\delta^{\text{in}}_{a_{x}} \rightarrow \infty$ the victims loss $L_v =  \delta_{v_x}$, as  $\lim_{\delta^{\text{in}}_{a_{x}} \rightarrow \infty}  \delta_{v_y} =0$. Thus, the bot's gain cannot exceed the victim's loss. 
\end{proof}
\subsection{Trader Perspective}\label{sec:trader}
Intending to minimize the victim's expected transaction execution cost, we turn to the victim's perspective of the sandwich game. We again consider an arbitrary victim's transaction $T_v=(\delta_{v_x},s,f,b,x_{0},y_{0},{t_0})$.  First, we note that we consider the victim's transaction unattackable for
$$ s \cdot \delta_{v_y}\geq 2 b,$$
as $s \cdot \delta_{v_y}$ an upper bound for the bot's profit (cf. Theorem~\ref{thm:profit}). Note that here the base fee for the transaction is given in currency $Y$. Thus, high slippage tolerances and trade sizes make victim trades attackable. Any $s \leq s_a$, where
$$ s_a = \frac{2b}{\delta_{v_y}},$$
ensure that no profitable sandwich attack for the victim's transaction exists. However, by selecting a low slippage tolerance, potential victims risk their trade failing to execute due to the natural movements in the pool, from trades, or liquidity withdrawals. Therefore, it is unreasonable to set a low slippage tolerance to avoid a sandwich attack when this low slippage tolerance is associated with high expected costs linked with resubmitting failed transactions.

The costs of transaction failure consist of the cost of redoing the transaction and the cost associated with the price shift between the two blocks. We estimate the cost of redoing the transaction to be
$(l+m) b$,
where $l$ is the portion of the base fee used for a failed transaction, and $m$ is the potential increase of the base fee in the next block. We set $l=0.25$, as the gas used by a failed Uniswap transaction is approximately a quarter of that of a successful transaction~\cite{2021ethertransaction}. Additionally, we set $m= 0.125$, as it is the maximum increase of the base fee within a block~\cite{2021eip1559invest}. The expected cost of the associated price shift in the pool is denoted by $\mathbb{E}(s|\tilde{s} >s)\delta_{v_y}$. More precisely, $-\mathbb{E}(s|\tilde{s} >s)$ is expected fractional price change given that the transaction failed ($ \tilde{s} >s$). Here, $\tilde{s}$ is the block's price slippage. Finally, we denote the probability of the transaction failing for slippage tolerance $s$ and trade size $\delta_{v_x}$ as $p(s,\delta_{v_x})$. Note, that $p(s,\delta_{v_x})$ can be estimated reliably by looking at the recent history of the pool (cf. Section~\ref{sec:predicts}). 

Thus, an approximative upper bound for redoing the transaction is given as
\begin{align*}
    \sum ^\infty _{i =1} p(s,\delta_{v_x})^i ((l+m) b+ \mathbb{E}(s|\tilde{s} >s)\delta_{v_y} ) \\
    = \frac{p(s,\delta_{v_x})}{1-p(s,\delta_{v_x})} ((l+m) b+ \mathbb{E}(s|\tilde{s} >s)\delta_{v_y} )
\end{align*}
Setting the slippage tolerance to $s<s_r$, where $$ s_r = \frac{p(s,\delta_{v_x})}{1-p(s,\delta_{v_x})} \left( \frac{(l+m) b}{\delta_{v_y}} + \mathbb{E}(s|\tilde{s} >s)\right),$$
ensures that the estimated costs associated with the transaction failing to execute do not exceed the costs of a possible sandwich attack. Note, that finding $s_r$ is possible with a ternary search, as the left side of the equation decreases with $s$, while the right side of the equation increases with $s$.

In case $s_r< s_a$, the potential victim can choose $s\in [s_r,s_a)$ to make sure that no profitable sandwich attack exists. We will always choose $s = s=s_a-\varepsilon$, where $\varepsilon\rightarrow 0^+$ to minimise the costs of transaction failure. Simultaneously, the victim does not face an unreasonable high expected cost related to the transaction failing. On the other hand, if $s_r \leq s_a$, the potential victim cannot easily set the slippage tolerance to avoid both sandwich attacks and the risk of having to pay the costs related to the transaction failing. However, as we find in Section~\ref{sec:settings}, this generally only occurs for comparatively large transactions. In reality, these transactions are better divided into several smaller trades to reduce their price impact. Price impact is an unrelated effect a trader should consider before executing a trade. 

We conclude the analysis by presenting the algorithm utilized by the trader to choose the optimal slippage tolerance in Algorithm~\ref{alg:slip}.

\begin{algorithm}[htpb]
    \centering
    \caption{Setting Slippage}\label{alg:slip}
    \begin{algorithmic}%[1]
        \State For transaction $T_v=(\delta_{v_x},s,f,b,x_{0},y_{0},{t_0})$ in pool $X\rightleftharpoons Y$
        \State Calculate $s_a = \frac{2b}{\delta_{v_y}}$ and  $s_r =\frac{p(s,\delta_{v_x})}{1-p(s,\delta_{v_x})} \left( \frac{(l+m) b}{\delta_{v_y}} + \mathbb{E}(s|\tilde{s} >s)\right)$ for transaction $T_v$
        \State \textbf{if} $s_r< s_a$: 
        \State \hspace{\algorithmicindent} set $s = s=s_a-\varepsilon$, where $\varepsilon\rightarrow 0^+$
        \State \textbf{else}: 
        \State \hspace{\algorithmicindent} set $s=s_r$ 
        
    \end{algorithmic}
\end{algorithm}
Algorithm~\ref{alg:slip} can also be used to set the slippage tolerance in Uniswap V3. The implementation of Algorithm~\ref{alg:slip} will vary only slightly between Uniswap V2 and V3. The estimations of both the probability of the transaction failing for slippage tolerance $s$ and trade size $\delta_{v_x}$ ($p(s,\delta_{v_x})$) and the expected fractional price change given that the transaction failed ($-\mathbb{E}(s|\tilde{s} >s)$) are calculated with the specific liquidity distribution. For Uniswap V2 (cf. Section~\ref{sec:predicts}) the number of tokens reserved in the pool suffice for the prediction.

\section{Evaluation}\label{sec:eval}

We analyze past Uniswap data to compare the costs for traders using the slippage tolerance proposed by Uniswap and the sandwich game. The data description follows in the succeeding section. 
\begin{table*}[htpb]
\npdecimalsign{.}

\nprounddigits{2}
\centering
\begin{tabular}{@{}rn{1}{2}n{1}{2}n{1}{2}n{1}{2}n{1}{2}n{1}{2}n{1}{2}n{1}{2}@{}}
\toprule
{} & \multicolumn{2}{c}{USDC$\rightleftharpoons$WETH} & \multicolumn{2}{c}{USDC$\rightleftharpoons$USDT} & \multicolumn{2}{c}{WBTC$\rightleftharpoons$WETH} & \multicolumn{2}{c}{DPI$\rightleftharpoons$WETH}\\
{} &  \multicolumn{1}{c}{$\mu$} &  \multicolumn{1}{c}{$\sigma$} &  \multicolumn{1}{c}{$\mu$} & \multicolumn{1}{c}{$\sigma$} &       \multicolumn{1}{c}{$\mu$} & \multicolumn{1}{c}{$\sigma$} & \multicolumn{1}{c}{$\mu$} & \multicolumn{1}{c}{$\sigma$} \\
\textbf{size [\$]}  &           &       &       &     &           &       &         &              \\
\midrule
10     & 1.8037E-4 & 6.1791E-3 & 9.5209E-5 & 8.3135E-4 & 6.8258E-5 & 9.2361E-4 & 1.6525E-4 & 1.1871E-3 \\
100    & 1.8069E-4 & 6.3544E-3 & 9.5208E-5 & 8.3114E-4 & 6.8306E-5 & 9.2481E-4 & 1.6526E-4 & 1.1870E-3 \\
1000   & 1.8169E-4 & 6.4532E-3 & 9.5203E-5 & 8.2963E-4 & 6.8737E-5 & 1.0698E-3 & 1.6533E-4 & 1.1876E-3 \\
10000  & 1.8406E-4 & 7.0676E-3 & 9.5136E-5 & 8.4849E-4 & 7.1852E-5 & 4.5681E-3 & 1.6553E-4 & 1.2321E-3 \\
100000 & 1.8493E-4 & 7.6691E-3 & 9.4190E-5 & 1.1519E-3 & 8.0817E-5 & 1.6770E-2 & 1.6329E-4 & 1.3943E-3 \\
\bottomrule
\end{tabular}

\caption{Mean ($\mu$) absolute fractional price change ($r$) and volatility ($\sigma$) of absolute fractional price change for four Uniswap pools: USDC$\rightleftharpoons$WETH, USDC$\rightleftharpoons$USDT,  WBTC$\rightleftharpoons$WETH and DPI$\rightleftharpoons$WETH.}\label{tab:meanvol}
\vspace{-16pt}
\end{table*}
\subsection{Data Description}
To collect data, we launch a go-ethereum client and export all transactions executed on Uniswap V2. We collect all Uniswap V2 transactions recorded on Ethereum up to block 11709847 (on 23 January 2021). In the following data analysis, we focus on 120,000 blocks (from block 11589848 to block 11709847) in January 2021, a particularly active time for Uniswap V2 before the launch of Uniswap V3. Thus, the trade activity on Uniswap V2 at this time is uninfluenced by Uniswap V3.\footnote{We note that while the data precedes flashbots, flashbots, however, does not impact a pool's price fluctuations but the success of sandwich attacks. As we assume optimal conditions for sandwich attacks anyways, this does not impact our analysis.}  We obtain the price of each cryptocurrency in a common currency, US\$ in our case, from the pool reserves and Coinbase~\cite{2021coinbase}.

In the following, we analyze data from eight Uniswap pools. The pools analyzed are USDC$\rightleftharpoons$WETH, USDC$\rightleftharpoons$USTD, WBTC$\rightleftharpoons$WETH, DPI$\rightleftharpoons$WETH, WBTC$\rightleftharpoons$USDC, UNI$\rightleftharpoons$USDC, LINK$\rightleftharpoons$WETH, and KIMCHI$\rightleftharpoons$WETH. We choose pools through a combination of size and type\footnote{Type divides pools into normal pools, stable pools, and exotic pools. These categories were introduced by Uniswap~\cite{adams2021uniswap}.} to represent a representative sample of Uniswap pools.

\subsection{Slippage Prediction}\label{sec:predicts}

To understand the price changes between blocks, we start by analyzing the fractional price change in all eight Uniswap pools over 120,000 blocks in January 2021. The absolute fractional price change ($r$) is given as:
$$r = \frac{\lvert\tilde{\delta}_{v_y}-\delta_{v_y} \rvert}{\delta_{v_y}}= \lvert s \rvert.$$
We see that the fractional price change is dependent on the trade size. Thus, we find the average absolute fractional price change and its volatility for five trade sizes (\$10, \$100, \$1000, \$10000, and \$100000) in  Table~\ref{tab:meanvol} for a selection of four pools. We note that we consider the anticipated trade output ($\delta_{v_y}$) to be the trade size throughout the entire analysis. These trade sizes cover the majority of trades executed on Uniswap -- Uniswap's median trade size was \$634 in 2020~\cite{2021unireview}.

\begin{table*}[htpb]
\npdecimalsign{.}

\nprounddigits{2}
\centering
\begin{subtable}{\textwidth}
\centering
\begin{tabular}{@{}rn{1}{2}rn{1}{2}rn{1}{2}rn{1}{2}r@{}}
\toprule
{} &  \multicolumn{2}{c}{USDC$\rightleftharpoons$WETH} & \multicolumn{2}{c}{USDC$\rightleftharpoons$USDT} & \multicolumn{2}{c}{WBTC$\rightleftharpoons$WETH} & \multicolumn{2}{c}{DPI$\rightleftharpoons$WETH}\\
{} &  \multicolumn{1}{c}{$\mu$} &  \multicolumn{1}{c}{$\eta$} &  \multicolumn{1}{c}{$\mu$} & \multicolumn{1}{c}{$\eta$} &       \multicolumn{1}{c}{$\mu$} & \multicolumn{1}{c}{$\eta$} & \multicolumn{1}{c}{$\mu$} & \multicolumn{1}{c}{$\eta$}\\
\textbf{window size }  &           &       &       &     &           &       &         &              \\
\midrule
200   & -2.3710E-3 & 0.637 & -8.0359E-4 & 0.512 & -1.0284E-3 & 0.611 & -1.6532E-3 & 0.656 \\
2000  & -2.7430E-3 & 0.093 & -8.9489E-4 & 0.065 & -1.2213E-3 & 0.106 & -2.0253E-3 & 0.078 \\
20000 & -2.9304E-3 & 0.014 & -9.2740E-4 & 0.014 & -1.3719E-3 & 0.007 & -2.1318E-3 & 0.045 \\
\bottomrule
\end{tabular}\caption{$p(s) =0.01$}
\end{subtable}\vspace{-2pt}

\begin{subtable}{\textwidth}
\centering
\begin{tabular}{@{}rn{1}{2}rn{1}{2}rn{1}{2}rn{1}{2}r@{}}
\toprule
{} &  \multicolumn{2}{c}{USDC$\rightleftharpoons$WETH} & \multicolumn{2}{c}{USDC$\rightleftharpoons$USDT} & \multicolumn{2}{c}{WBTC$\rightleftharpoons$WETH} & \multicolumn{2}{c}{DPI$\rightleftharpoons$WETH}\\
{} &  \multicolumn{1}{c}{$\mu$} &  \multicolumn{1}{c}{$\eta$} &  \multicolumn{1}{c}{$\mu$} & \multicolumn{1}{c}{$\eta$} &       \multicolumn{1}{c}{$\mu$} & \multicolumn{1}{c}{$\eta$} & \multicolumn{1}{c}{$\mu$} & \multicolumn{1}{c}{$\eta$}\\
\textbf{window size }  &           &       &       &     &           &       &         &              \\
\midrule
200   & -9.2175E-4 & 0.124 & -9.0532E-5 & 0.024 & -1.4719E-4 & 0.063 & -2.6085E-4 & 0.063 \\
2000  & -9.7416E-4 & 0.013 & -7.7557E-5 & 0.021 & -1.0560E-4 & 0.022 & -1.9000E-4 & 0.025 \\
20000 & -9.8791E-4 & 0.007 & -8.3898E-5 & 0.019 & -7.8738E-5 & 0.020 & -1.5188E-4 & 0.018 \\
\bottomrule
\end{tabular}\caption{$p(s) =0.05$}
\end{subtable}\vspace{-2pt}

\begin{subtable}{\textwidth}
\centering
\begin{tabular}{@{}rn{1}{2}rn{1}{2}rn{1}{2}rn{1}{2}r@{}}
\toprule
{} &  \multicolumn{2}{c}{USDC$\rightleftharpoons$WETH} & \multicolumn{2}{c}{USDC$\rightleftharpoons$USDT} & \multicolumn{2}{c}{WBTC$\rightleftharpoons$WETH} & \multicolumn{2}{c}{DPI$\rightleftharpoons$WETH}\\
{} &  \multicolumn{1}{c}{$\mu$} &  \multicolumn{1}{c}{$\eta$} &  \multicolumn{1}{c}{$\mu$} & \multicolumn{1}{c}{$\eta$} &       \multicolumn{1}{c}{$\mu$} & \multicolumn{1}{c}{$\eta$} & \multicolumn{1}{c}{$\mu$} & \multicolumn{1}{c}{$\eta$}\\
\textbf{window size }  &           &       &       &     &           &       &         &              \\
\midrule
200   & -3.4878E-4 & 0.042 & -7.3481E-6 & 0.335 & -1.8537E-5 & 0.194 & -4.3575E-5 & 0.213 \\
2000  & -2.9881E-4 & 0.001 & -1.2389E-6 & 0.314 & -4.3408E-6 & 0.148 & -2.1846E-5 & 0.186 \\
20000 & -2.5581E-4 & 0.003 &  0.0000 & 0.310 & -1.0433E-6 & 0.114 & -7.8062E-6 & 0.143 \\
\bottomrule
\end{tabular}\caption{$p(s) =0.1$}\label{tab:acc3}
\end{subtable}\vspace{-6pt}
\caption{Average ($\mu$) and relative error ($\eta$) of slippage tolerance $s$ prediction using historical percentile for transaction failure probabilities $p(s) \in [0.01,0.05,0.1]$ and window sizes $w\in [200,2000,20000]$ for four Uniswap pools: USDC$\rightleftharpoons$WETH, USDC$\rightleftharpoons$USDT,  WBTC$\rightleftharpoons$WETH and DPI$\rightleftharpoons$WETH.}\label{tab:acc}
\vspace{-16pt}
\end{table*}
We notice immediately that the mean absolute fractional price change is small in all considered pools -- contradicting the common assumption that the price of cryptocurrencies fluctuates significantly, even between blocks. Instead, we find the price to be relatively constant between two blocks (around 13 seconds). Further, the average absolute price change is significantly less than the fractional slippage tolerance of $5 \cdot 10^{-3}$ proposed by Uniswap across all four pools~\cite{2021uniinterface}. The difference is even more startling as slippage only concerns negative price changes. 

In the sandwich game, the trader estimates the required slippage tolerance $s$ such that the probability of the transaction failing is $p(s)$. To allow facile computation, we estimate the required slippage tolerance $\hat{s}_{p(s)}^w$ such that the probability of transaction failure is $p(s)$ to be the $p(s)^{\text{th}}$ percentile of the observed fractional price change in the past $w$ blocks. Here, $w$ is the window size used for the estimation. We then compute the accuracy of our estimation over 120,000 blocks and summarize the results in Table~\ref{tab:acc}. There we show the mean ($\mu$) and the relative error ($\eta$) of the prediction of $\hat{s}_{p(s)}^w$ for a given the probability of transaction failure $p(s)$ and window size $w$. 

\begin{figure*}[htbp]
\centering

    \begin{subfigure}{0.47\linewidth}
        \includegraphics[width = \textwidth]{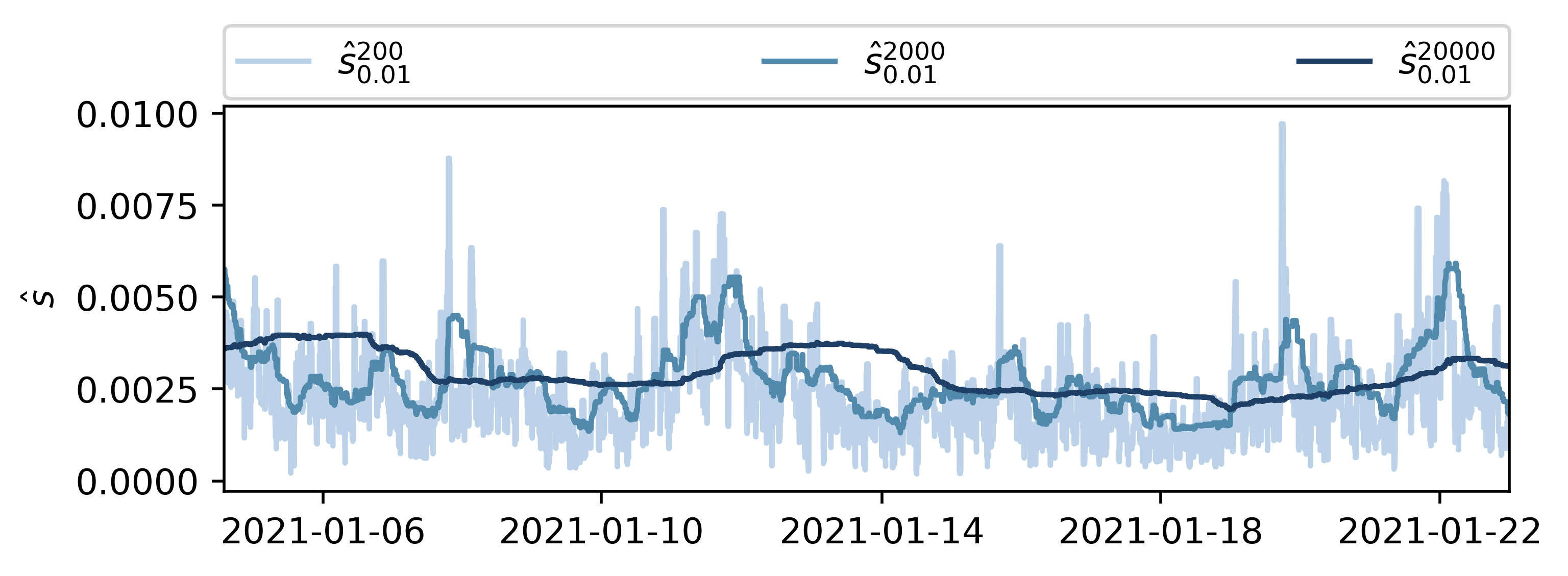}\vspace{-6pt}
        \caption{USDC$\rightleftharpoons$ETH} \label{fig:UW01}
    \end{subfigure}%
    \hfill
    \begin{subfigure}{0.47\linewidth}
        \includegraphics[width = \textwidth]{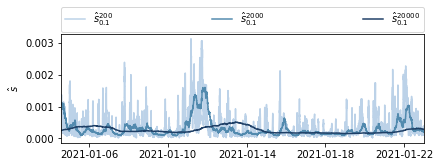}\vspace{-6pt}
        \caption{USDC$\rightleftharpoons$ETH} \label{fig:UW1}
    \end{subfigure}%
    
    \begin{subfigure}{0.47\linewidth}
        \includegraphics[width = \textwidth]{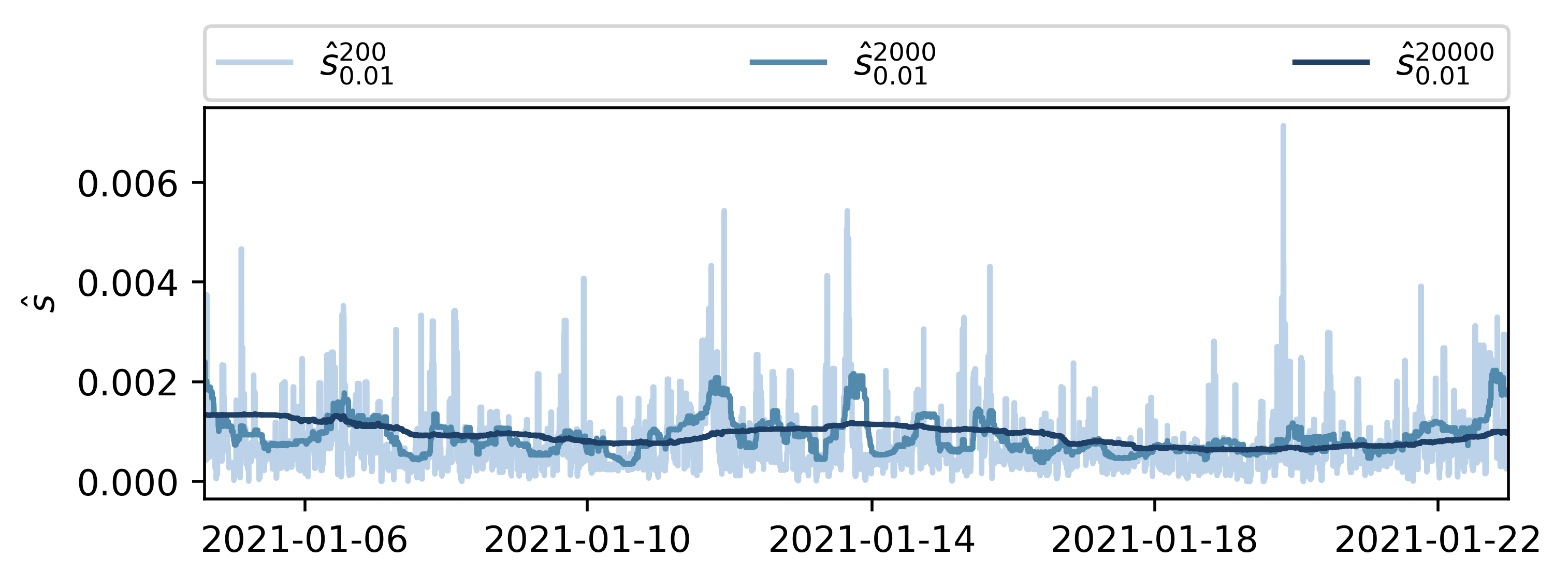}\vspace{-6pt}
        \caption{USDC$\rightleftharpoons$USDT} \label{fig:UU01}
    \end{subfigure}%
    \hfill
    \begin{subfigure}{0.47\linewidth}
        \includegraphics[width = \textwidth]{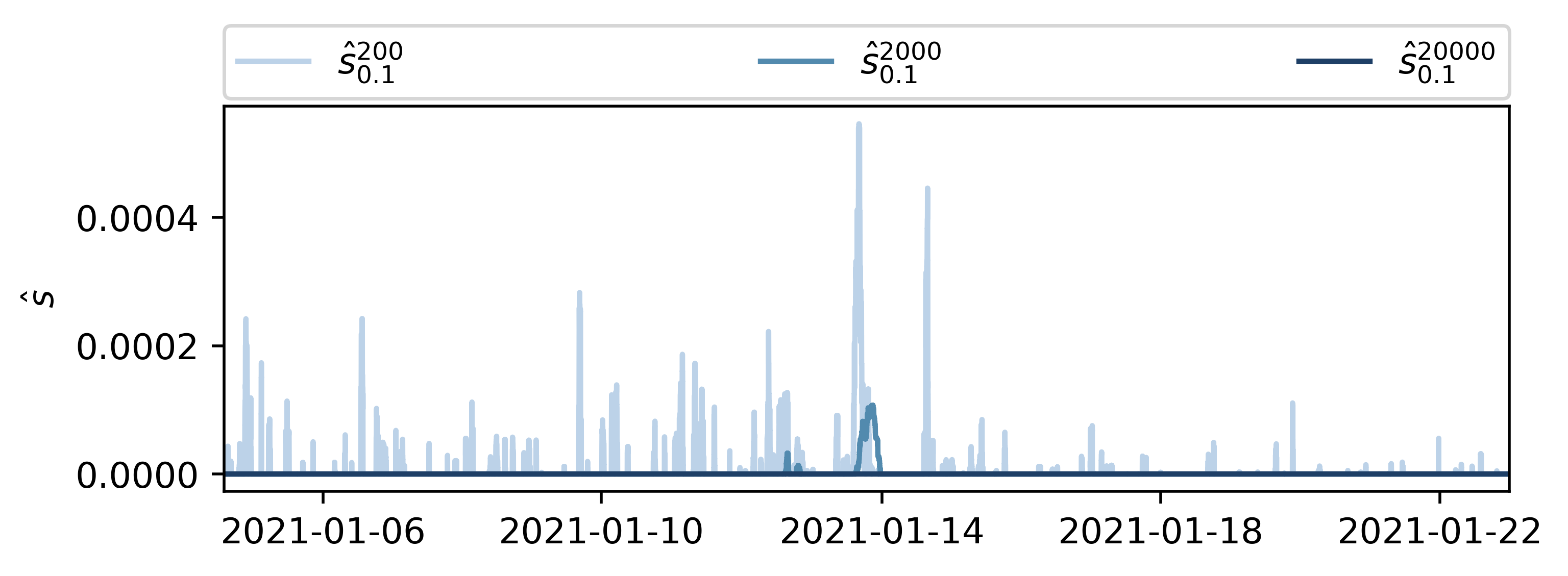}\vspace{-6pt}
        \caption{USDC$\rightleftharpoons$USDT} \label{fig:UU1}
    \end{subfigure}%
    
    \begin{subfigure}{0.47\linewidth}
        \includegraphics[width = \textwidth]{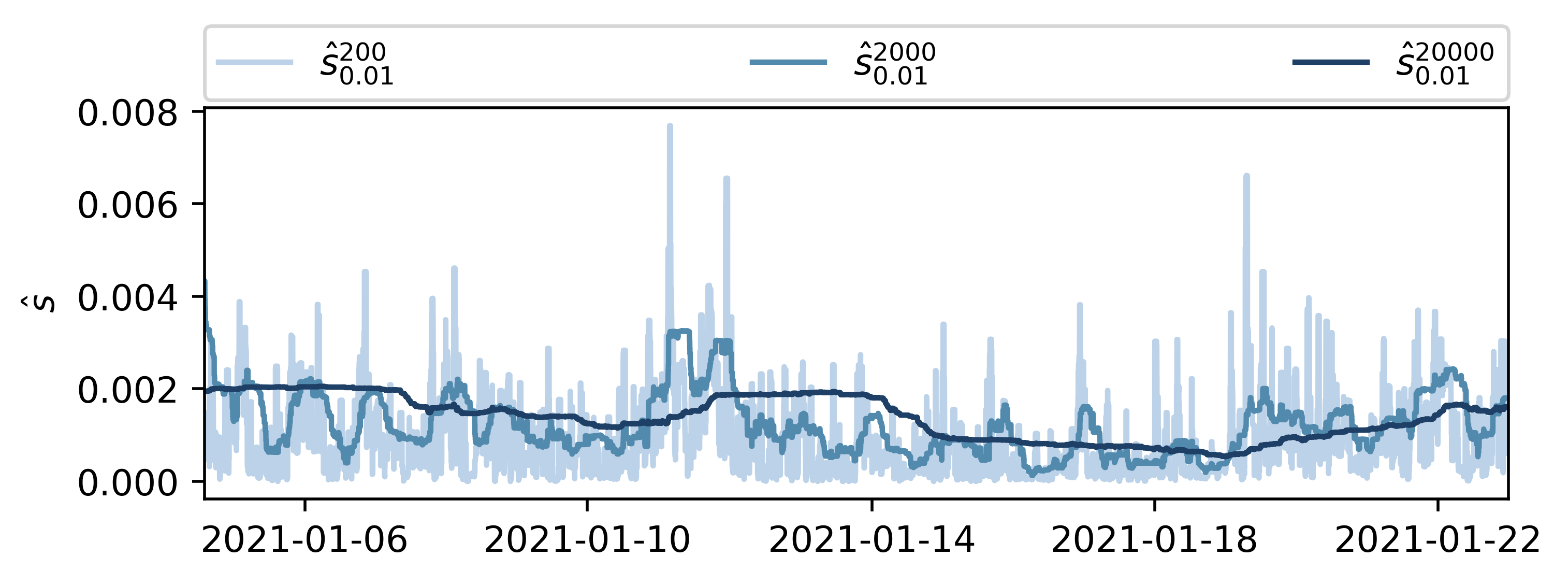}\vspace{-6pt}
        \caption{BTC$\rightleftharpoons$ETH} \label{fig:BU01}
    \end{subfigure}%
    \hfill
    \begin{subfigure}{0.47\linewidth}
        \includegraphics[width = \textwidth]{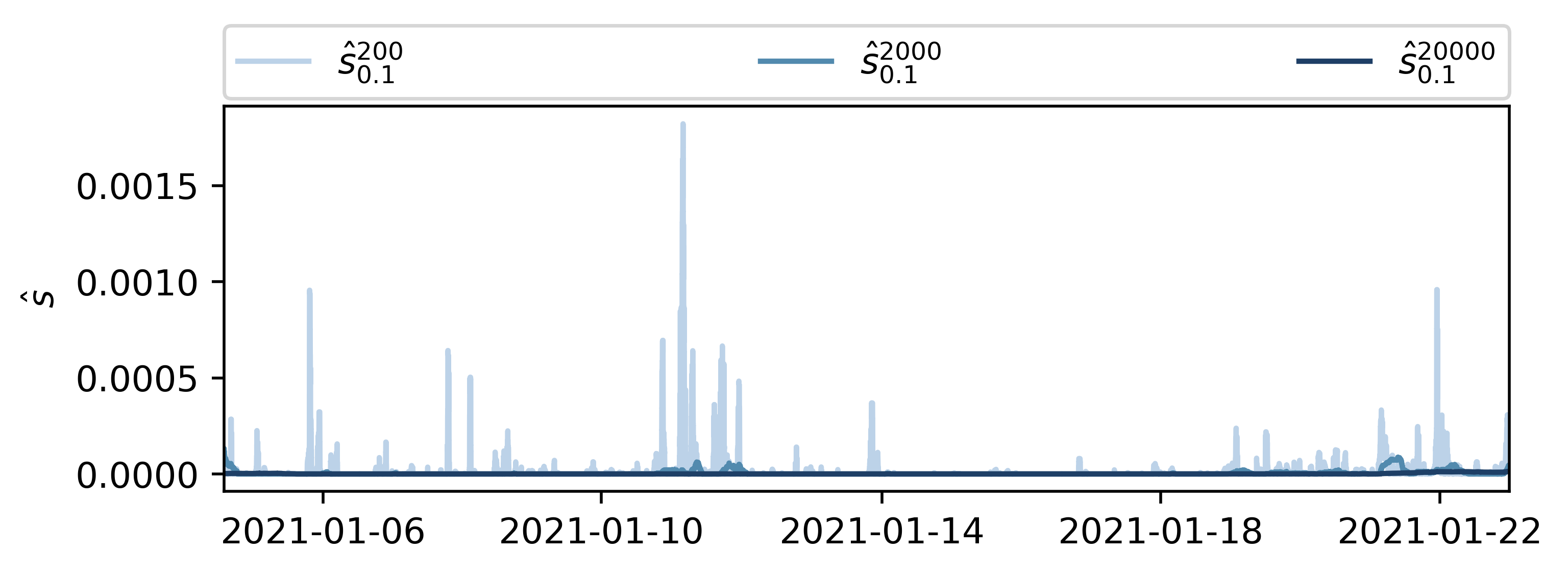}\vspace{-6pt}
        \caption{BTC$\rightleftharpoons$ETH} \label{fig:BU1}
    \end{subfigure}%
    
    \begin{subfigure}{0.47\linewidth}
        \includegraphics[width = \textwidth]{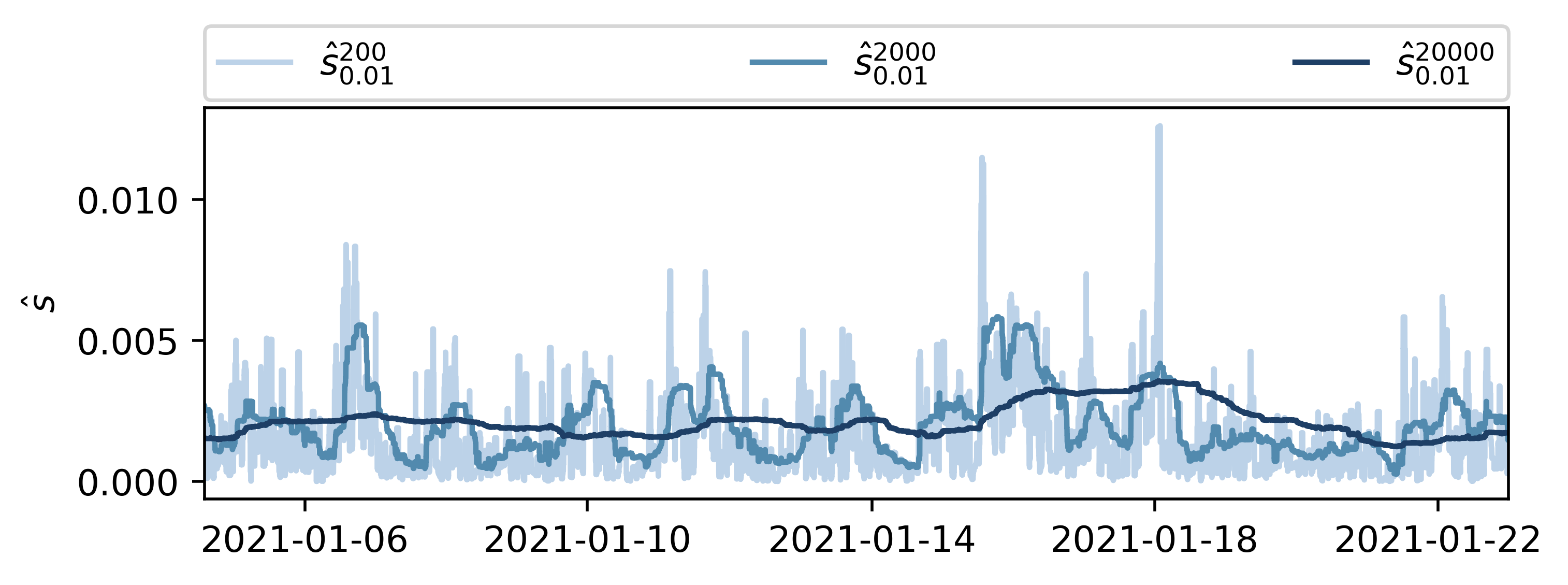}\vspace{-6pt}
        \caption{DPI$\rightleftharpoons$ETH} \label{fig:DE01}
    \end{subfigure}%
    \hfill
    \begin{subfigure}{0.47\linewidth}
        \includegraphics[width = \textwidth]{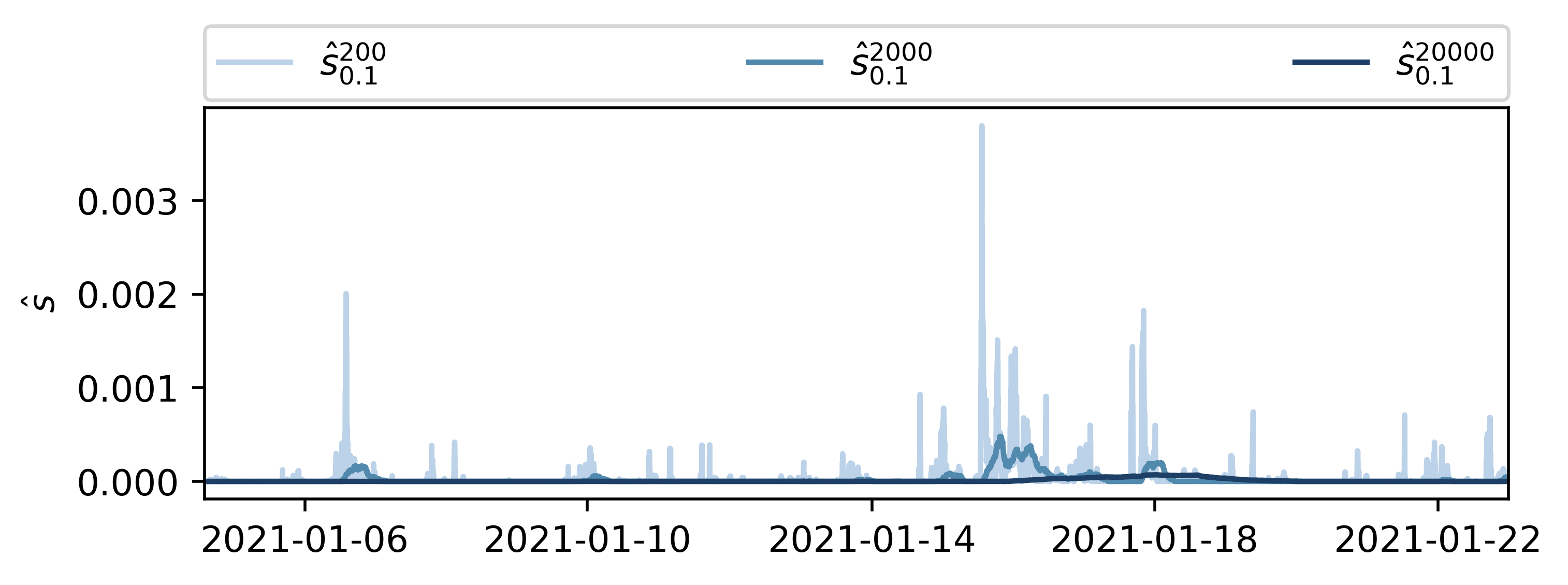}\vspace{-6pt}
        \caption{DPI$\rightleftharpoons$ETH} \label{fig:DE1}
    \end{subfigure}\vspace{-10pt}
\caption{Required slippage prediction $\left(\hat{s}_{p(s)}^w\right)$ for transaction failure probabilities $p(s)\in[0.01,0.05,0.1]$ and window sizes $w\in [200,2000,20000]$ for four Uniswap pools: USDC$\rightleftharpoons$WETH, USDC$\rightleftharpoons$USDT,  WBTC$\rightleftharpoons$WETH and DPI$\rightleftharpoons$WETH. We predict the required slippage tolerance over 120,000 blocks from block 11589848 to block 11709847.}\label{fig:prediction} 
\vspace{-10pt}
\end{figure*}

While the approximation is largely inaccurate for the smallest window size ($w=200$), it is accurate for all larger window sizes. Only for the largest tested slippage tolerance (cf. Table~\ref{tab:acc3}) does the prediction become inaccurate. However, this stems from the probability of transaction failure $p(s)=0.1$ being large enough, such that in less than a fraction of $p(s)$ blocks, the fractional price change is positive. Consequently, the estimation $\hat{s}_{p(s)}^w$ becomes zero. This is true for three of the tested pools: WBTC$\rightleftharpoons$USDC (cf. Figure~\ref{fig:BU1}), USDC$\rightleftharpoons$USDT (cf. Figure~\ref{fig:UU1}) and DPI$\rightleftharpoons$WETH (cf. Figure~\ref{fig:DE1}), and caused by low volume in the pools. If no trade is executed in the pool during a block, the required slippage tolerance is inevitably zero. Only for the most active pool (USDC$\rightleftharpoons$WETH), does prediction  $\hat{s}_{p(s)}^w$ remain accurate for $p(s) =0.1$ (cf. Figure~\ref{fig:UW1}). However, as these inaccuracies cause the slippage tolerance to be over-estimated rather than under-estimated, they do not cause unnecessary failures. In the following, we will use $w = 2000$ as a window size for the estimation. The estimation does not become noticeably more accurate for larger window sizes, and using  $w = 2000$  allows the system to react to changes more quickly.
\subsection{Setting Slippage}\label{sec:settings}

\begin{figure}[htpb]
\centering
    \begin{subfigure}{\linewidth}
        \includegraphics[width = 0.98\textwidth]{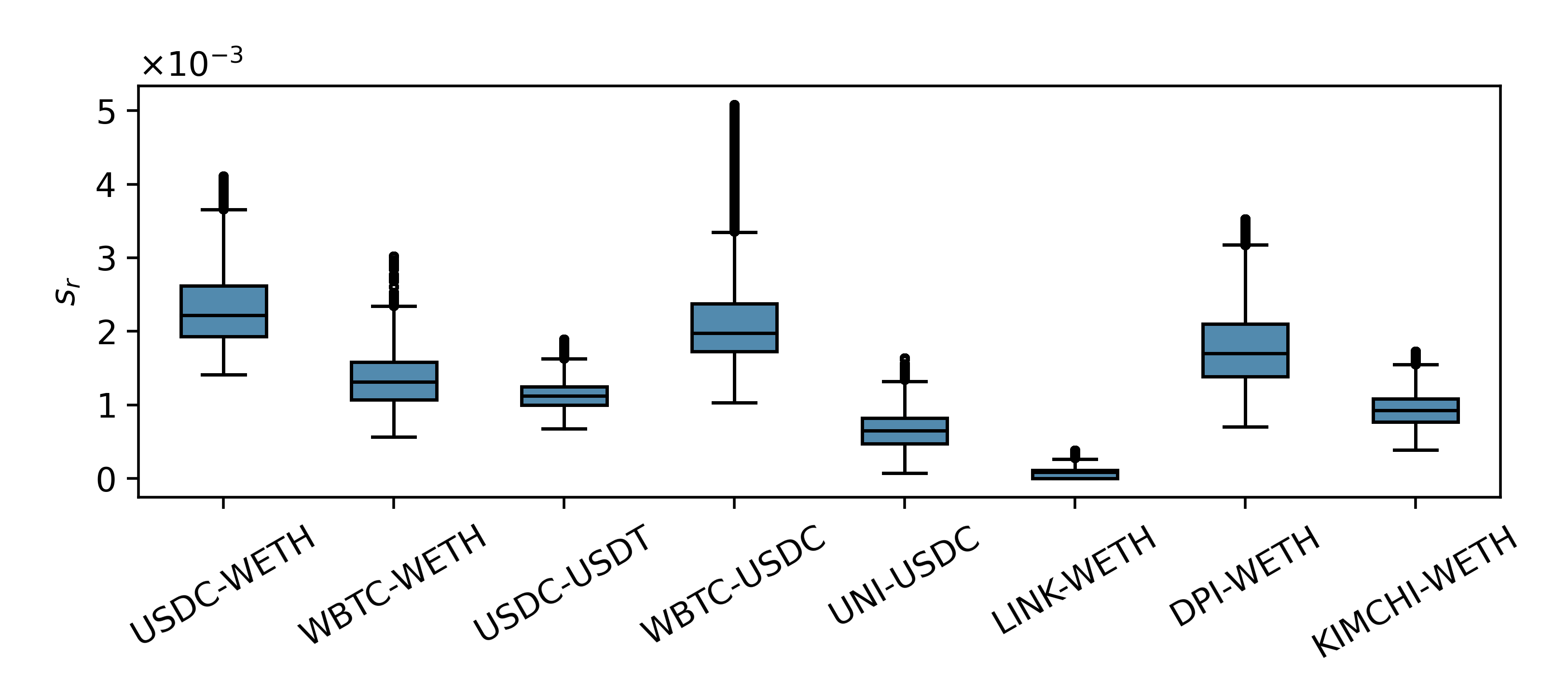}\vspace{-10pt}
        \caption{transaction size: $\delta_{v_y}=\$10$} \label{fig:10s}
    \end{subfigure}
    \begin{subfigure}{\linewidth}
        \includegraphics[width = 0.98\textwidth]{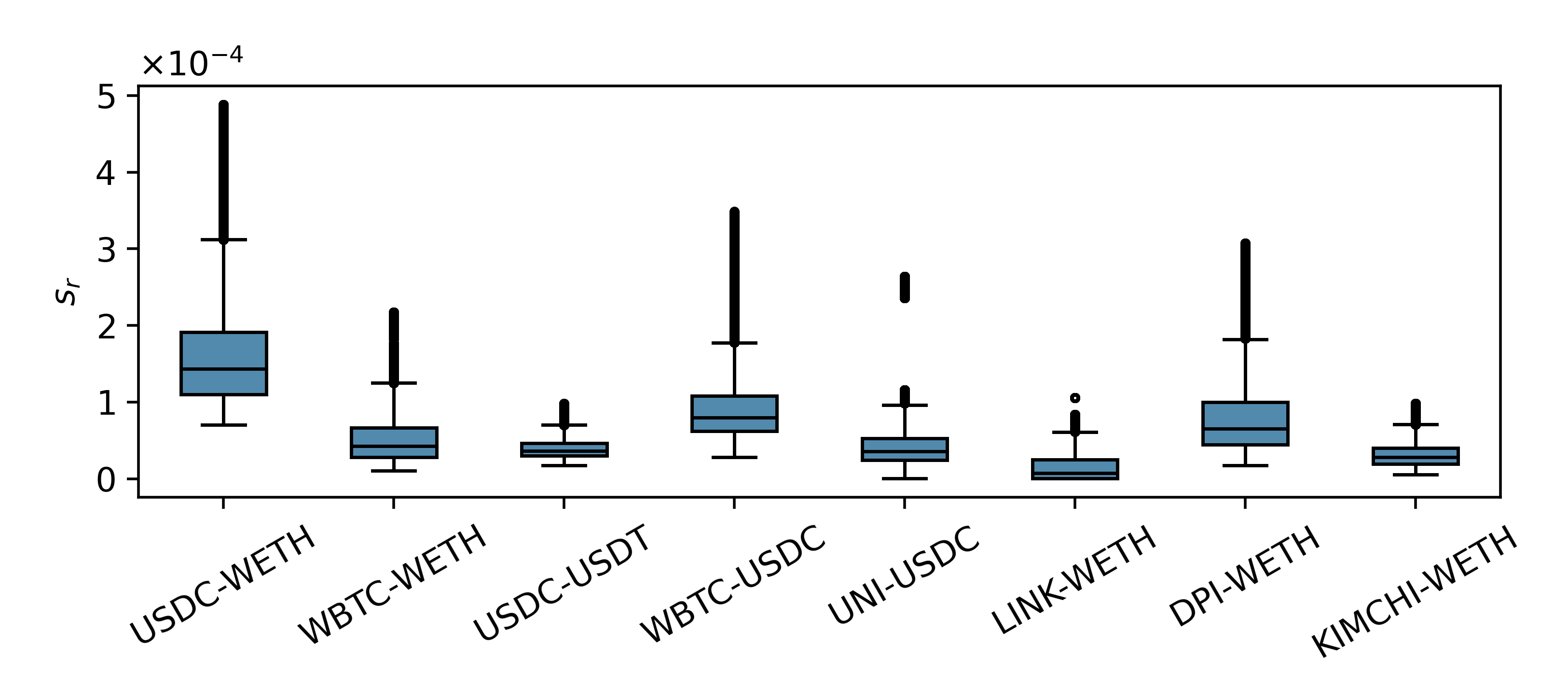}\vspace{-10pt}
        \caption{transaction size: $\delta_{v_y}=\$100000$} \label{fig:100000s}
    \end{subfigure}\vspace{-6pt}
\caption{We compute the lower bound for the slippage tolerance ($s_r$) (Algorithm~\ref{alg:slip}) over 120,000 blocks from block 11589848 to block 11709847. $s_r$ adapts to current pool characteristics and we, thus, record different values for every block which we show as a blue boxplot for pools: USDC$\rightleftharpoons$WETH, USDC$\rightleftharpoons$USTD, WBTC$\rightleftharpoons$WETH, DPI$\rightleftharpoons$WETH, WBTC$\rightleftharpoons$USDC,
UNI$\rightleftharpoons$USDC, LINK$\rightleftharpoons$WETH, and KIMCHI$\rightleftharpoons$WETH.}\label{fig:sets}
\vspace{-10pt}
\end{figure}
With the ability to predict the required slippage tolerance, we continue by calculating the slippage tolerance's lower bound, ensuring that the expected cost of transaction failure does not exceed the cost of a sandwich attack. We find this lower bound for trades of sizes: \$10, \$100, \$1000, \$10000, and \$100000 using Algorithm~\ref{alg:slip} for each block in our data set. In the following evaluation, we set the base fee to \$4. A base fee of \$4 for a Uniswap V2 transaction is in line with current values~\cite{2021burn}. We repeat our evaluation with different base fees (\$2 and \$8) in Appendix~\ref{app:base}. Over 120,000 blocks, we compute the lower bound for the slippage tolerance ($s_r$)  for the eight analyzed pools. As $s_r$ adapts to the current pool characteristics, we compute different values for every block. We visualize the results as a box plot for two trade sizes (\$10 and \$100000) in Figure~\ref{fig:sets}.

Note that even though the trade size differs by a factor of 10000, $s_r$ only decreases by a factor of 10 (cf. Figures~\ref{fig:10s} and~\ref{fig:100000s}). Further, we note that for both trade sizes, we observe a similar pattern between pools. $s_r$ tends to be smaller for pools with lower volume such as LINK$\rightleftharpoons$WETH and is largest for USDC$\rightleftharpoons$WETH, the biggest pool in terms of volume. This trend might be counter-intuitive initially, as we would expect prices of these, generally more exotic, cryptocurrencies in lower volume pools to fluctuate more. However, while this might be true for larger time frames, e.g., days, this is not true in the time-scale of blocks. Due to the low trading volume in the pools, there are many blocks without any trade execution. Thus, there are no price fluctuations between these blocks. We also see that $s_r$ differs within pools across time. For instance, we observe that $s_r$ varies by a factor of more than five for USDC$\rightleftharpoons$WETH for $\delta_{v_y}=\$100000$. Pools go through periods of both lower and higher volume. Therefore, it is natural that the expected fractional price change between two blocks also varies over time. Observing the difference of $s_r$ within and across pools indicates that the constant auto-slippage, as suggested by several AMMs, cannot be suitable for all trades. We will further underscore this point in the following with a comparison of the slippage tolerances computed by Algorithm~\ref{alg:slip} and Uniswap's constant auto-slippage (cf. Figure~\ref{fig:compares}).

\begin{figure}[htpb]
\centering
    \begin{subfigure}{\linewidth}
        \includegraphics[width = \textwidth]{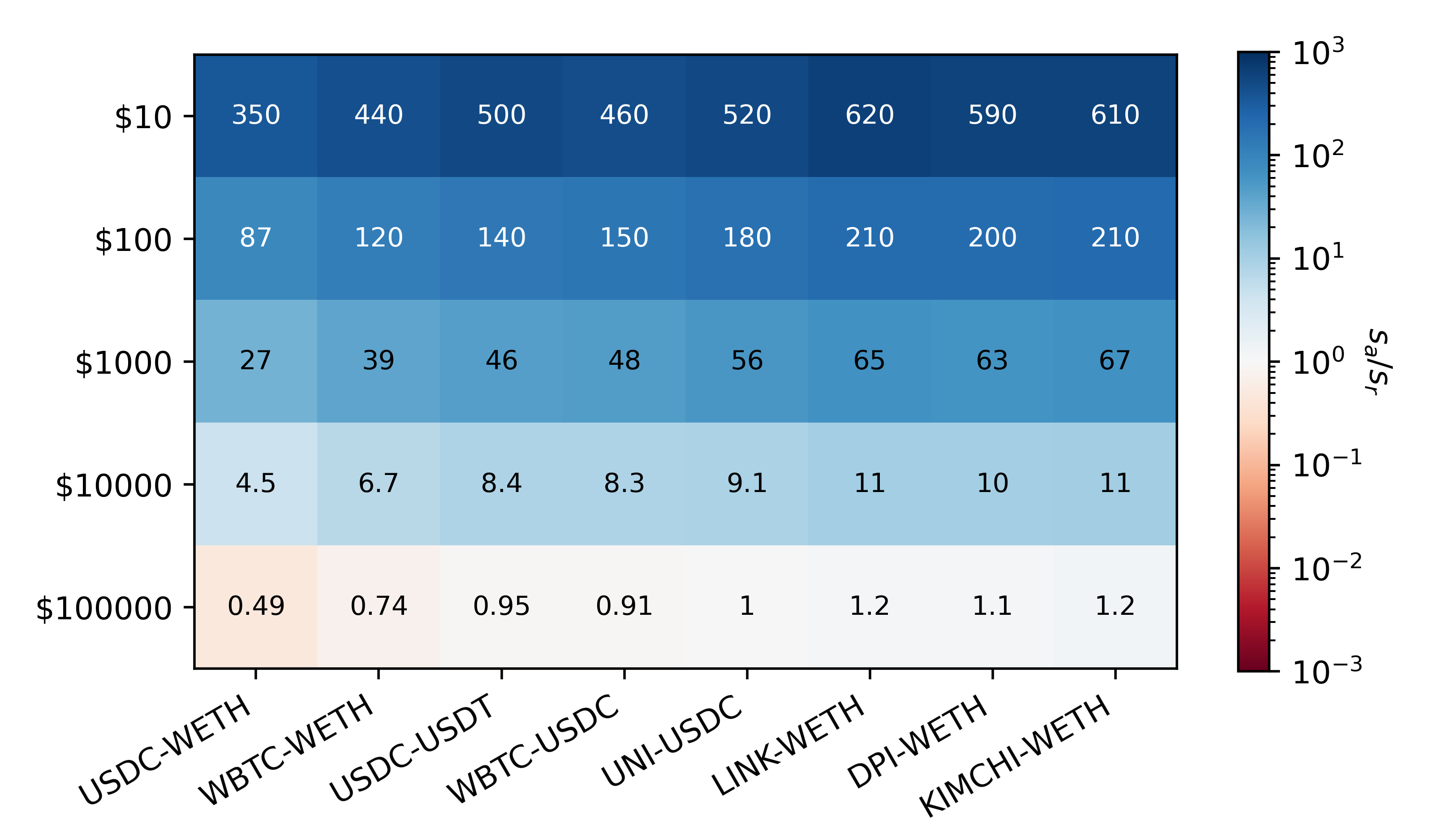}\vspace{-6pt}
        \caption{Comparison between the slippage tolerance at which trades become attackable ($s_a$) and the mean of the lower bound for the slippage tolerance such that the expected costs of transaction failure does not exceed the cost of a sandwich attack ($s_r$). Values larger than 1 suggest that sandwich attacks can be avoided ($s=s_a$ in Algorithm~\ref{alg:slip}), while values smaller than 1 indicate that sandwich attacks cannot be avoided easily ($s=s_r$ in Algorithm~\ref{alg:slip}).} \label{fig:sasr}
    \end{subfigure}%

    \begin{subfigure}{\linewidth}
        \includegraphics[width = \textwidth]{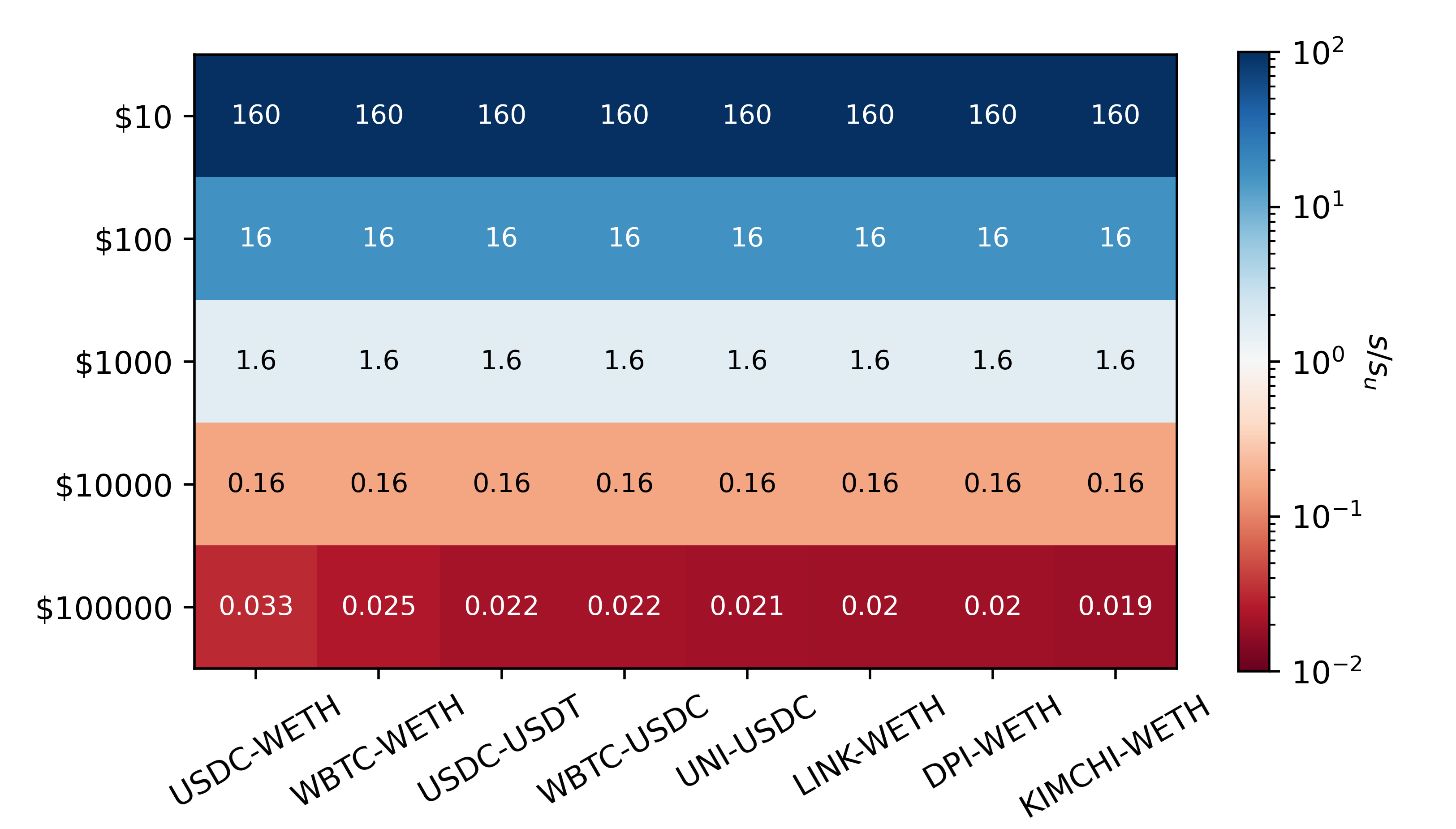}\vspace{-6pt}
        \caption{Comparison between the slippage tolerance chosen by Algorithm~\ref{alg:slip} ($s$) and the auto-slippage suggested by Uniswap ($s_u$). Values larger than 1 suggest that Uniswap auto-slippage is to low thereby leads to unnecessary trade failures. On the other hand, values larger than 1 indicate that the auto-slippage suggested by Uniswap is too high and unnecessary sandwich attacks occur.} \label{fig:oursu}
    \end{subfigure}\vspace{-4pt}
\caption{Slippage tolerance for pools: USDC$\rightleftharpoons$WETH, USDC$\rightleftharpoons$USTD, WBTC$\rightleftharpoons$WETH, DPI$\rightleftharpoons$WETH, WBTC$\rightleftharpoons$ USDC, UNI$\rightleftharpoons$USDC, LINK$\rightleftharpoons$WETH, and KIMCHI$\rightleftharpoons$WETH and trade sizes: \$10, \$100, \$1000, \$10000 and \$100000.}\label{fig:compares}
\vspace{-14pt}
\end{figure}

In Figure~\ref{fig:sasr} we compare $s_a$ and $s_r$. We find that the mean value of $s_r$ does not exceed $s_a$ for all transaction sizes analyzed up to \$10000. Note that when looking at the entire data set,  $s_r$ never exceeds $s_a$ for these transaction sizes. Thus, for all these transactions, the slippage tolerance can easily be set to $s=s_a-\varepsilon$, $\varepsilon\rightarrow 0^+$, to avoid being attacked and ensure that the costs related to potentially having to redo the transaction are small. Only for the largest transaction size does $s_r$ occasionally exceed $s_a$. The mean value of $s_r$ exceeds $s_a$ in half the pools and in Figure~\ref{fig:100000s} we see that there is at least one block for all pools in which $s_r$ surpasses $s_a$. Thus, when the trade size exceeds \$100000, sandwich attacks are not (always) avoidable with our parameter configuration. 

We turn to Figure~\ref{fig:oursu}, where we compare the slippage tolerance chosen by Algorithm~\ref{alg:slip} ($s$) to the slippage tolerance recommended by Uniswap ($s_\text{u}$). We show in blue where $s_u$ is smaller than $s$ and in red where $s_u$ exceeds $s$. For small trade sizes, $s_u$ is comparatively small. This unnecessarily small slippage tolerance is up to a factor of 160 smaller than the slippage tolerance at which the trade becomes attackable ($s_a$) and causes easily preventable transaction failures. We notice that independent of the transaction size, Uniswap's interface warns users that their transaction may be front-run when setting the slippage tolerance slightly below $s_a$ for trades of size \$10 and \$100. As $s_a$ specifies the slippage tolerance at which trades become attackable, they cannot be front-run profitably. Thus, the warning is misleading and can cause to unnecessary transaction failures. We note that while $s_a$ depends on the current base fee, Uniswap's warnings are fixed and independent of trade size, pool, and base fee. Thus, it suffices to test the Uniswap interface with realistic base fees. 

Simultaneously, for large trades, $s_\text{u}$ exceeds $s$ by up to a factor of more than 50, and thus opens greater parts of the transaction up for attacks than necessary. For example, when setting the slippage tolerance as indicated by our algorithm for trades of size \$10000, Uniswap warns that the transaction may fail and suggests users use a higher slippage tolerance. While not necessarily incorrect, any transaction may fail, the warning might encourage users to choose a higher slippage tolerance. Consequently, these users would encounter excess costs, as we will show in the subsequent section.

\subsection{Cost Comparison}\label{sec:cost}
To conclude the analysis, we simulate trades of sizes \$10, \$100, \$1000, \$10000 and \$100000 in every block between blocks  11589848 and 11709847 across all eight pools. We simulate all trades both with the slippage tolerance as specified by Algorithm~\ref{alg:slip} and with the slippage tolerance suggested by Uniswap. We note that we consider a trade $T_v=(\delta_{v_x},s,f,b,x_{0},y_{0},{t_0})$ to be attackable, whenever $s\delta_{v_y}\geq2b$ in accordance with Theorem~\ref{thm:profit}. 

We summarize the results of the simulation in Table~\ref{tab:costsaved}, where we show the fractional cost incurred when using our algorithm and the cost incurred when using the slippage tolerance recommended by Uniswap. This cost consists of both of the cost incurred from sandwich attacks and of the costs involved in resubmitting failed transactions. We further provide the detailed results on the number of times transactions fail to execute and suffer sandwich attacks in Appendix~\ref{app:details}. 

Our algorithm is significantly more cost-effective than the suggestions from Uniswap for all analyzed trade sizes in all analyzed pools. We notice that across all pools, very small trades experience no additional costs in our case but fail from time to time with Uniswap's suggested slippage tolerance. As we saw in Figure~\ref{fig:oursu}, the transactions fail as Uniswap's constant slippage tolerance is unnecessarily low for small trade sizes and leads to easily avoidable trade failures. While trades of size \$10 are never attacked nor fail when utilizing our algorithm for setting the slippage tolerance, a couple of transactions always fail when using Uniswap's slippage tolerance suggestion -- leading to an infinite cost reduction.

While our protocol for setting the slippage tolerance still saves costs in comparison to the auto-slippage across all pools, we find the smallest difference in costs for trades of size \$1000. We infer that the auto-slippage selected by Uniswap appears reasonable for transactions of size \$1000 when the base fee is \$4. This finding is in line with our observations from Figure~\ref{fig:oursu}, $s_u$ is closest to the slippage tolerance suggested by Algorithm~\ref{alg:slip} for trades of size \$1000.

\nprounddigits{3}
\newcolumntype{P}[1]{>{\centering\arraybackslash}p{#1}}
\begin{table*}[htpb]
\centering
\begin{subtable}{0.49\linewidth}
\centering
\begin{tabular}{@{}rn{1}{3}n{1}{3}rr@{}}
\toprule
\textbf{size [\$]} &   \multicolumn{1}{@{}P{2cm}@{}}{\textbf{fractional cost ours}} &  \multicolumn{1}{@{}P{1.9cm}@{}}{\textbf{fractional cost UNI}} &\multicolumn{1}{@{}P{1.8cm}@{}}{\textbf{ratio cost UNI/ours}} \\
\midrule
    10 & 0.0000 & 2.2670E-4 &    $\infty$\\
   100 & 0.0000 & 3.5450E-5 &    $\infty$\\
  1000 & 3.5548E-6 & 1.6325E-5 &  4.5924 \\
 10000 & 1.4347E-4 & 5.1037E-3 & 35.5718 \\
100000 & 3.1785E-4 & 5.0137E-3 & 15.7735 \\
\bottomrule
\end{tabular}
\caption{USDC$\rightleftharpoons$WETH}\label{tab:costsavedUE}
\end{subtable}
\hfill
\begin{subtable}{0.49\linewidth}
\centering
\begin{tabular}{@{}rn{1}{3}n{1}{3}rr@{}}
\toprule
\textbf{size [\$]} &   \multicolumn{1}{@{}P{2cm}@{}}{\textbf{fractional cost ours}} &  \multicolumn{1}{@{}P{1.9cm}@{}}{\textbf{fractional cost UNI}} &\multicolumn{1}{@{}P{1.8cm}@{}}{\textbf{ratio cost UNI/ours}} \\
\midrule
    10 & 0.0000 & 7.4405E-5 &     $\infty$\\
   100 & 2.4902E-6 & 1.5155E-5 &   6.0858 \\
  1000 & 5.8299E-6 & 9.2297E-6 &   1.5832 \\
 10000 & 4.1326E-5 & 5.1053E-3 & 123.5364 \\
100000 & 6.5758E-5 & 5.0153E-3 &  76.2684 \\
\bottomrule
\end{tabular}
\caption{WBTC$\rightleftharpoons$WETH}\label{tab:costsavedBE}
\end{subtable}

\begin{subtable}{0.49\linewidth}
\centering
\begin{tabular}{@{}rn{1}{3}n{1}{3}rr@{}}
\toprule
\textbf{size [\$]} &   \multicolumn{1}{@{}P{2cm}@{}}{\textbf{fractional cost ours}} &  \multicolumn{1}{@{}P{1.9cm}@{}}{\textbf{fractional cost UNI}} &\multicolumn{1}{@{}P{1.8cm}@{}}{\textbf{ratio cost UNI/ours}} \\
\midrule
    10 & 0.0000 & 8.3107E-5 &     $\infty$\\
   100 & 0.0000 & 1.3357E-5 &     $\infty$\\
  1000 & 2.0864E-6 & 6.3817E-6 &   3.0588 \\
 10000 & 2.6127E-5 & 5.1017E-3 & 195.2647 \\
100000 & 4.1509E-5 & 5.0117E-3 & 120.7390 \\
\bottomrule
\end{tabular}
\caption{USDC$\rightleftharpoons$USDT}\label{tab:costsavedUU}
\end{subtable}
\hfill
\begin{subtable}{0.49\linewidth}
\centering
\begin{tabular}{@{}rn{1}{3}n{1}{3}rr@{}}
\toprule
\textbf{size [\$]} &   \multicolumn{1}{@{}P{2cm}@{}}{\textbf{fractional cost ours}} &  \multicolumn{1}{@{}P{1.9cm}@{}}{\textbf{fractional cost UNI}} &\multicolumn{1}{@{}P{1.8cm}@{}}{\textbf{ratio cost UNI/ours}} \\
\midrule
    10 & 0.0000 & 8.0255E-4 &    $\infty$\\
   100 & 0.0000 & 1.3430E-4 &    $\infty$\\
  1000 & 3.4748E-5 & 6.7471E-5 &  1.9417 \\
 10000 & 9.7636E-5 & 5.1233E-3 & 52.4730 \\
100000 & 1.6187E-4 & 5.0325E-3 & 31.0901 \\
\bottomrule
\end{tabular}
\caption{WBTC$\rightleftharpoons$USDC}\label{tab:costsavedBU}
\end{subtable}

\begin{subtable}{0.49\linewidth}
\centering
\begin{tabular}{@{}rn{1}{3}n{1}{3}rr@{}}
\toprule
\textbf{size [\$]} &   \multicolumn{1}{@{}P{2cm}@{}}{\textbf{fractional cost ours}} &  \multicolumn{1}{@{}P{1.9cm}@{}}{\textbf{fractional cost UNI}} &\multicolumn{1}{@{}P{1.8cm}@{}}{\textbf{ratio cost UNI/ours}} \\
\midrule
    10 & 0.0000 & 3.2072E-4 &     $\infty$\\
   100 & 0.0000 & 7.6060E-5 &     $\infty$\\
  1000 & 4.7096E-5 & 5.1469E-5 &   1.0929 \\
 10000 & 5.0798E-5 & 5.1333E-3 & 101.0540 \\
100000 & 5.0975E-5 & 5.0363E-3 &  98.7995 \\
\bottomrule
\end{tabular}
\caption{UNI$\rightleftharpoons$USDC}\label{tab:costsavedUNU}
\end{subtable}
\hfill
\begin{subtable}{0.49\linewidth}
\centering
\begin{tabular}{@{}rn{1}{3}n{1}{3}rr@{}}
\toprule
\textbf{size [\$]} &   \multicolumn{1}{@{}P{2cm}@{}}{\textbf{fractional cost ours}} &  \multicolumn{1}{@{}P{1.9cm}@{}}{\textbf{fractional cost UNI}} &\multicolumn{1}{@{}P{1.8cm}@{}}{\textbf{ratio cost UNI/ours}} \\
\midrule
    10 & 0.0000 & 5.7070E-5 &     $\infty$\\
   100 & 4.4708E-6 & 2.0320E-5 &   4.5450 \\
  1000 & 1.6592E-5 & 1.6644E-5 &   1.0031 \\
 10000 & 1.6373E-5 & 5.1142E-3 & 312.3494 \\
100000 & 1.8341E-5 & 5.0241E-3 & 273.9272 \\
\bottomrule
\end{tabular}
\caption{LINK$\rightleftharpoons$WETH}\label{tab:costsavedLE}
\end{subtable}

\begin{subtable}{0.49\linewidth}
\centering
\begin{tabular}{@{}rn{1}{3}n{1}{3}rr@{}}
\toprule
\textbf{size [\$]} &   \multicolumn{1}{@{}P{2cm}@{}}{\textbf{fractional cost ours}} &  \multicolumn{1}{@{}P{1.9cm}@{}}{\textbf{fractional cost UNI}} &\multicolumn{1}{@{}P{1.8cm}@{}}{\textbf{ratio cost UNI/ours}} \\
\midrule
    10 & 0.0000 & 2.7638E-4 &    $\infty$\\
   100 & 2.4676E-6 & 4.6882E-5 & 18.9989 \\
  1000 & 9.2089E-6 & 2.3932E-5 &  2.5988 \\
 10000 & 7.2344E-5 & 5.1586E-3 & 71.3064 \\
100000 & 1.3238E-4 & 5.0237E-3 & 37.9494 \\
\bottomrule
\end{tabular}
\caption{DPI$\rightleftharpoons$WETH}\label{tab:costsavedDE}
\end{subtable}
\hfill
\begin{subtable}{0.49\linewidth}
\centering
\begin{tabular}{@{}rn{1}{3}n{1}{3}rr@{}}
\toprule
\textbf{size [\$]} &   \multicolumn{1}{@{}P{2cm}@{}}{\textbf{fractional cost ours}} &  \multicolumn{1}{@{}P{1.9cm}@{}}{\textbf{fractional cost UNI}} &\multicolumn{1}{@{}P{1.8cm}@{}}{\textbf{ratio cost UNI/ours}} \\
\midrule
    10 & 0.0000 & 2.7638E-4 &    $\infty$\\
   100 & 2.4676E-6 & 4.6882E-5 & 18.9989 \\
  1000 & 9.2089E-6 & 2.3932E-5 &  2.5988 \\
 10000 & 7.2321E-5 & 5.1087E-3 & 70.6393 \\
100000 & 1.3198E-4 & 5.0186E-3 & 38.0253 \\
\bottomrule
\end{tabular}
\caption{KIMCHI$\rightleftharpoons$WETH}\label{tab:costsavedKE}
\end{subtable}
\vspace{-6pt}
\caption{Cost comparison when using our own algorithm to set the slippage tolerance vs. the slippage tolerance suggested by Uniswap. The fractional cost includes both the costs of being attacked as well as the costs associated with redoing the transactions. The simulation spans over 120,000 blocks, from block 11589848 to block 11709847, and the base fee is set to \$4.}\label{tab:costsaved}
\vspace{-14pt}
\end{table*}

Finally, for large trades (\$10000 and \$100000), our algorithm consistently demonstrates a high cost reduction of up to a factor of 273. Looking at the results in further detail, we observe differing patterns for high volume pools such as USDC$\rightleftharpoons$WETH (cf. Table~\ref{tab:costsavedUE}) and WBTC$\rightleftharpoons$WETH (cf. Table~\ref{tab:costsavedBE}) and lower volume pools such as UNI$\rightleftharpoons$USDC (cf. Table~\ref{tab:costsavedUNU}) and LINK$\rightleftharpoons$WETH (cf. Table~\ref{tab:costsavedLE}). In high volume pools, the difference between the costs experienced by trades using our algorithm and Uniswap's auto-slippage decreases more starkly for large trades. Regardless of this decrease, the difference remains significant across all pools. In comparatively low volume pools, the cost ratio does not decrease noticeably for large trades. Low volume leads to smaller inter-block price movements: allowing our algorithm to select lower slippage tolerances and avoid sandwich attacks. Precisely, while 80\% of trades using our slippage tolerance algorithm were attacked for trades of size \$100000 in the USDC$\rightleftharpoons$WETH (cf. Table~\ref{tab:extraUE}), less than 3\% of trades were attacked in the low volume pool LINK$\rightleftharpoons$WETH (cf. Table~\ref{tab:extraLE}). 

Thus, we deduce that using a constant auto-slippage, as suggested by both Uniswap and SushiSwap, ignorant of the trade size and pool characteristics, imposes unreasonably high costs on trades. The inefficiency of the constant auto-slippage is highlighted by our algorithm repeatedly demonstrating a cost reduction of a three-figure factor. Further, we note that setting the slippage tolerance per our simple algorithm avoids sandwich attacks for all tested pools and transaction sizes smaller than \$100000. This success shows that contrary to common assumptions, traders can mostly avoid being sandwich attacked by setting the slippage tolerance. 

To conclude, we infer that in pools with smaller inter-block price movements, the additional costs traders need to face from the transaction ordering tax can be reduced significantly. In Uniswap V3, liquidity providers no longer automatically commit to providing liquidity for the entire price range but can choose to provide liquidity in a smaller price range~\cite{adams2021uniswap}. As a consequence, their liquidity is up to 4000 times more capitally efficient~\cite{2021uniintro}. Thus, we expect inter-block price movements to be even smaller, and our algorithm would allow traders to avoid the invisible tax even further.

\section{Related Work}
The prevalence of front-running on centralized exchanges is a well-studied area~\cite{BERNHARDT2008front,angel2011equity} and most types of front-running are outlawed in traditional markets~\cite{markham1988front,moosa2015regulation}. Still, there are legal trading strategies utilized by high-frequency trading (HTF) firms that front-run transactions for profit~\cite{harris2013what,scopino2014the}. 

Only with the introduction of Ethereum DApps has front-running become a pervasive issue on permissionless blockchains. Eskandir et al.~\cite{eskandari2019sok} are the first to combine the scattered body of knowledge of front-running on permissionless blockchains at the time. Seeing the effects of front-running on AMM users and the limited actions taken by the AMMs themselves, we offer them a simple way of protecting themselves against such attacks. 

Daian et al.~\cite{daian2020flash} present a study on \textit{price gas auctions} (PGA), analyzing various types of predatory trading behaviors known from traditional finance and adapting to DeFi. They further introduce \textit{miner-extractable value} (MEV) as a concept and empirically show its risks. MEV measures the profit miners can extract through either arbitrarily including or excluding transactions from blocks or re-ordering transactions within blocks. Subsequently, Qin et al.~\cite{qin2021quantifying} quantify the transaction ordering tax and provide evidence of miners already extracting MEV. In contrast, we focus specifically on sandwich attacks from both the victims' and bot's perspectives by introducing the sandwich game.  

Zhou et al.~\cite{zhou2021high} formalize the sandwich attack problem on AMM exchanges. They study the problem analytically and empirically from the attackers' perspective and quantify when profitable attacks exist. We generalize the analytical sandwich attack problem and include the victim perspective -- letting victims adjust the slippage tolerance to avoid sandwich attacks. Our analysis reveals that contrary to popular belief, victims can mitigate sandwich attacks in most cases.

A large-scale analysis of sandwich attacks is performed by Züst in~\cite{zust2021analyzing}: quantifying the frequency and profitability of sandwich attacks and showing that the number of bots performing sandwich attacks is becoming increasingly efficient. While Züst suggests splitting up large trades as a mitigation strategy, we demonstrate that it is generally sufficient for DeFi users to adjust their slippage tolerance to protect against sandwich attacks. 

Several solutions to blockchain front-running have been introduced recently. With Tesseract, Bentov et al.~\cite{bentov2019tesseract} introduce an exchange that relies on trusted hardware to resit front-running. Aequitas is a premissioned consensus protocol to achieve order-fairness by Kelkar et al.~\cite{kelkar2020order}. Cachin et al.~\cite{cachin2021quick} strengthen the fairness notion achieved by Aequitas. In contrast to these works, we show that sandwich attacks are preventable without the need for trusted hardware or premissioned consensus. Further, our approach allows users to protect themselves immediately without having to wait for the DeFi ecosystem to evolve.

\section{Conclusion}
%%
%% The next two lines define the bibliography style to be used, and
%% the bibliography file.

Sandwich attacks are a constant threat to the transactions of traders on AMMs. In this work, we generalized the sandwich attack problem to include both traders and bots. Our model demonstrates that the constant auto-slippage suggested by most AMMs only performs well for a small set of trade parameters. Further, we highlight that, contrary to popular belief, traders can easily avoid most sandwich attacks. An adjustment of the slippage tolerance suffices in most cases and does not face an unnecessarily high risk of trade failure due to an insufficiently small slippage tolerance. The simple algorithm we present can be utilized by traders to protect themselves against sandwich attacks and outperforms the auto-slippage suggested by Uniswap in all tested settings -- demonstrating a three-figure factor cost reduction. We foresee the possibility that some more conservative traders prefer accepting the transaction ordering tax instead of accepting the small risk of transaction failure. However, this would open up the opportunity for AMMs themselves or a new DeFi service to guarantee a given (low) slippage tolerance to their users by amortizing the cost across a pool of users. 

While our simple approach is successful at avoiding sandwich attacks without incurring unnecessary costs and allows traders to protect themselves, it does not prevent other predatory trading behaviors leading to MEV. The development of an approach to prevent all predatory trading behaviors is, thus, an open question for future research.
\bibliographystyle{ACM-Reference-Format} 
\bibliography{ccs}

%%
%% If your work has an appendix, this is the place to put it.
\newpage
\appendix
\section{Failed and Attacked Trades}\label{app:details}

\begin{table*}[!b]
\centering
\begin{subtable}{0.49\linewidth}
\centering
\begin{tabular}{@{}rrrrrrrr@{}}
\toprule
{} & \multicolumn{2}{@{}P{2.1cm}@{}}{\textbf{failed trades}} & \multicolumn{2}{@{}P{2.4cm}@{}}{\textbf{average failed attempts}} & \multicolumn{2}{@{}P{2.4cm}@{}}{\textbf{attacked trades}} \\
{} &       \multicolumn{1}{c}{ours} &   \multicolumn{1}{c}{UNI} &   \multicolumn{1}{c}{ours} & \multicolumn{1}{c}{UNI} &       \multicolumn{1}{c}{ours} &   \multicolumn{1}{c}{UNI} \\
\textbf{size [\$]}  &           &       &       &     &           &  \\
\midrule
10     &     0 & 253 &    0.0000 & 1.0079 &       0 &      0 \\
100    &     0 & 253 &    0.0000 & 1.0079 &       0 &      0 \\
1000   &    36 & 253 &    1.0000 & 1.0079 &       0 &      0 \\
10000  &  6814 & 253 &    1.1611 & 1.0079 &       0 & 119747 \\
100000 & 14697 & 253 &    1.2232 & 1.0079 &  101371 & 119747 \\
\bottomrule
\end{tabular}
\caption{USDC$\rightleftharpoons$WETH}\label{tab:extraUE}

\end{subtable}
\hfill
\begin{subtable}{0.49\linewidth}
\centering
\begin{tabular}{@{}rrrrrrrr@{}}
\toprule
{} & \multicolumn{2}{@{}P{2.1cm}@{}}{\textbf{failed trades}} & \multicolumn{2}{@{}P{2.4cm}@{}}{\textbf{average failed attempts}} & \multicolumn{2}{@{}P{2.4cm}@{}}{\textbf{attacked trades}} \\
{} &       \multicolumn{1}{c}{ours} &   \multicolumn{1}{c}{UNI} &   \multicolumn{1}{c}{ours} & \multicolumn{1}{c}{UNI} &       \multicolumn{1}{c}{ours} &   \multicolumn{1}{c}{UNI} \\
\textbf{size [\$]}  &           &       &       &     &           &  \\
\midrule
10     &     0 &  79 &    0.0000 & 1.0000 &       0 &      0 \\
100    &     3 &  79 &    1.0000 & 1.0000 &       0 &      0 \\
1000   &    21 &  79 &    1.0000 & 1.0000 &       0 &      0 \\
10000  &  1992 &  79 &    1.0658 & 1.0000 &       0 & 119921 \\
100000 &  5455 &  79 &    1.0948 & 1.0000 &   16026 & 119921 \\
\bottomrule
\end{tabular}
\caption{WBTC$\rightleftharpoons$WETH}
\end{subtable}
\begin{subtable}{0.49\linewidth}
\centering
\begin{tabular}{@{}rrrrrrrr@{}}
\toprule
{} & \multicolumn{2}{@{}P{2.1cm}@{}}{\textbf{failed trades}} & \multicolumn{2}{@{}P{2.4cm}@{}}{\textbf{average failed attempts}} & \multicolumn{2}{@{}P{2.4cm}@{}}{\textbf{attacked trades}} \\
{} &       \multicolumn{1}{c}{ours} &   \multicolumn{1}{c}{UNI} &   \multicolumn{1}{c}{ours} & \multicolumn{1}{c}{UNI} &       \multicolumn{1}{c}{ours} &   \multicolumn{1}{c}{UNI} \\
\textbf{size [\$]}  &           &       &       &     &           &  \\
\midrule
10     &     0 &  93 &    0.0000 & 1.0000 &       0 &      0 \\
100    &     0 &  93 &    0.0000 & 1.0000 &       0 &      0 \\
1000   &    20 &  93 &    1.0000 & 1.0000 &       0 &      0 \\
10000  &  1336 &  93 &    1.0352 & 1.0000 &       0 & 119907 \\
100000 &  5964 &  92 &    1.1160 & 1.0000 &    1536 & 119908 \\
\bottomrule
\end{tabular}
\caption{USDC$\rightleftharpoons$USDT}

\end{subtable}
\hfill
\begin{subtable}{0.49\linewidth}
\centering
\begin{tabular}{@{}rrrrrrrr@{}}
\toprule
{} & \multicolumn{2}{@{}P{2.1cm}@{}}{\textbf{failed trades}} & \multicolumn{2}{@{}P{2.4cm}@{}}{\textbf{average failed attempts}} & \multicolumn{2}{@{}P{2.4cm}@{}}{\textbf{attacked trades}} \\
{} &       \multicolumn{1}{c}{ours} &   \multicolumn{1}{c}{UNI} &   \multicolumn{1}{c}{ours} & \multicolumn{1}{c}{UNI} &       \multicolumn{1}{c}{ours} &   \multicolumn{1}{c}{UNI} \\
\textbf{size [\$]}  &           &       &       &     &           &  \\
\midrule
10     &     0 & 881 &    0.0000 & 1.0114 &       0 &      0 \\
100    &     0 & 881 &    0.0000 & 1.0114 &       0 &      0 \\
1000   &   345 & 881 &    1.0087 & 1.0114 &       0 &      0 \\
10000  &  2450 & 880 &    1.0351 & 1.0114 &       0 & 119120 \\
100000 &  3189 & 858 &    1.0442 & 1.0105 &   57470 & 119142 \\
\bottomrule
\end{tabular}
\caption{WBTC$\rightleftharpoons$USDC}
\end{subtable}
\begin{subtable}{0.49\linewidth}
\centering
\begin{tabular}{@{}rrrrrrrr@{}}
\toprule
{} & \multicolumn{2}{@{}P{2.1cm}@{}}{\textbf{failed trades}} & \multicolumn{2}{@{}P{2.4cm}@{}}{\textbf{average failed attempts}} & \multicolumn{2}{@{}P{2.4cm}@{}}{\textbf{attacked trades}} \\
{} &       \multicolumn{1}{c}{ours} &   \multicolumn{1}{c}{UNI} &   \multicolumn{1}{c}{ours} & \multicolumn{1}{c}{UNI} &       \multicolumn{1}{c}{ours} &   \multicolumn{1}{c}{UNI} \\
\textbf{size [\$]}  &           &       &       &     &           &  \\
\midrule
10     &     0 & 325 &    0.0000 & 1.0031 &       0 &      0 \\
100    &     0 & 324 &    0.0000 & 1.0031 &       0 &      0 \\
1000   &   252 & 325 &    1.0040 & 1.0031 &       0 &      0 \\
10000  &   482 & 317 &    1.0041 & 1.0032 &       0 & 119683 \\
100000 &   542 & 277 &    1.0037 & 1.0036 &    7872 & 119723 \\
\bottomrule
\end{tabular}
\caption{UNI$\rightleftharpoons$USDC}

\end{subtable}
\hfill
\begin{subtable}{0.49\linewidth}
\centering
\begin{tabular}{@{}rrrrrrrr@{}}
\toprule
{} & \multicolumn{2}{@{}P{2.1cm}@{}}{\textbf{failed trades}} & \multicolumn{2}{@{}P{2.4cm}@{}}{\textbf{average failed attempts}} & \multicolumn{2}{@{}P{2.4cm}@{}}{\textbf{attacked trades}} \\
{} &       \multicolumn{1}{c}{ours} &   \multicolumn{1}{c}{UNI} &   \multicolumn{1}{c}{ours} & \multicolumn{1}{c}{UNI} &       \multicolumn{1}{c}{ours} &   \multicolumn{1}{c}{UNI} \\
\textbf{size [\$]}  &           &       &       &     &           &  \\
\midrule
10     &     0 &  49 &    0.0000 & 1.0000 &       0 &      0 \\
100    &     5 &  49 &    1.0000 & 1.0000 &       0 &      0 \\
1000   &    48 &  49 &    1.0000 & 1.0000 &       0 &      0 \\
10000  &    52 &  49 &    1.0000 & 1.0000 &       0 & 119951 \\
100000 &    52 &  49 &    1.0000 & 1.0000 &    3080 & 119951 \\
\bottomrule
\end{tabular}
\caption{LINK$\rightleftharpoons$WETH}\label{tab:extraLE}

\end{subtable}
\begin{subtable}{0.49\linewidth}
\centering
\begin{tabular}{@{}rrrrrrrr@{}}
\toprule
{} & \multicolumn{2}{@{}P{2.1cm}@{}}{\textbf{failed trades}} & \multicolumn{2}{@{}P{2.4cm}@{}}{\textbf{average failed attempts}} & \multicolumn{2}{@{}P{2.4cm}@{}}{\textbf{attacked trades}} \\
{} &       \multicolumn{1}{c}{ours} &   \multicolumn{1}{c}{UNI} &   \multicolumn{1}{c}{ours} & \multicolumn{1}{c}{UNI} &       \multicolumn{1}{c}{ours} &   \multicolumn{1}{c}{UNI} \\
\textbf{size [\$]}  &           &       &       &     &           &  \\
\midrule
10     &     0 & 305 &    0.0000 & 1.0033 &       0 &      0 \\
100    &     3 & 305 &    1.0000 & 1.0033 &       0 &      0 \\
1000   &    56 & 305 &    1.0000 & 1.0033 &       0 &      0 \\
10000  &  2720 & 305 &    1.1162 & 1.0033 &       0 & 119695 \\
100000 &  6701 & 304 &    1.1310 & 1.0033 &   41511 & 119696 \\
\bottomrule
\end{tabular}
\caption{DPI$\rightleftharpoons$WETH}

\end{subtable}
\hfill
\begin{subtable}{0.49\linewidth}
\centering
\begin{tabular}{@{}rrrrrrrr@{}}
\toprule
{} & \multicolumn{2}{@{}P{2.1cm}@{}}{\textbf{failed trades}} & \multicolumn{2}{@{}P{2.4cm}@{}}{\textbf{average failed attempts}} & \multicolumn{2}{@{}P{2.4cm}@{}}{\textbf{attacked trades}} \\
{} &       \multicolumn{1}{c}{ours} &   \multicolumn{1}{c}{UNI} &   \multicolumn{1}{c}{ours} & \multicolumn{1}{c}{UNI} &       \multicolumn{1}{c}{ours} &   \multicolumn{1}{c}{UNI} \\
\textbf{size [\$]}  &           &       &       &     &           &  \\
\midrule
10     &     0 & 183 &    0.0000 & 1.0055 &       0 &      0 \\
100    &     3 & 183 &    1.0000 & 1.0055 &       0 &      0 \\
1000   &    57 & 183 &    1.0000 & 1.0055 &       0 &      0 \\
10000  &   815 & 183 &    1.0049 & 1.0055 &       0 & 119817 \\
100000 &  2591 & 183 &    1.0243 & 1.0055 &    1706 & 119817 \\
\bottomrule
\end{tabular}
\caption{KIMCHI$\rightleftharpoons$WETH}\label{tab:extraKE}
\end{subtable}\vspace{-6pt}
\caption{Comparison between the number of failed and attacked trades when using our algorithm to set the slippage tolerance vs. the slippage tolerance suggested by Uniswap. The simulation spans between blocks 11589848 and 1170984. The base fee is set to \$4.}\label{tab:fails}\vspace{-14pt}
\end{table*}
Table~\ref{tab:fails} compares the number of failed and attacked trades when using our slippage tolerance setting algorithm and Uniswap's suggested slippage tolerance. It is apparent while the constant slippage tolerance used by Uniswap cannot be effective for both different trade patterns and pools. Our simple slippage tolerance algorithm, on the other hand, adjusts well to the varying conditions. We notice that Uniswap's auto-slippage leads to a similar number of failed trades caused by an insufficient slippage tolerance for all tested trade sizes in a pool. Especially for smaller trade sizes, these failures are unnecessary, as we see when comparing the performance of the auto-slippage to that chosen by our slippage tolerance algorithm. For trade sizes up to a \$1000,  our algorithm can successfully avoid the vast majority of trade failures in all pools, and simultaneously not a single trade suffers a sandwich attack. All large (\$10000 and \$100000) trades with Uniswap's auto-slippage are sandwich attacked or failed to execute. Our algorithm, on the other hand, avoids all sandwich attacks for trades up to size \$1000, while at the same time only experiencing a few, at most 5.5\% for trades of size \$10000 in the USDC$\rightleftharpoons$WETH pool (cf. Table~\ref{tab:extraUE}), trade failures. For the largest trade size (\$100000), our algorithm cannot avoid all sandwich attacks. Still, trades are only attacked very rarely in comparison to those that use Uniswap's auto-slippage.
\section{Cost Comparison \texorpdfstring{($b=\$2$ and $b=\$8$)}{} }\label{app:base}

\nprounddigits{3}
\begin{table*}[!b]
\centering
\begin{subtable}{0.49\linewidth}
\centering
\begin{tabular}{@{}rn{1}{3}n{1}{3}rr@{}}
\toprule
\textbf{size [\$]} &   \multicolumn{1}{@{}P{2cm}@{}}{\textbf{fractional cost ours}} &  \multicolumn{1}{@{}P{1.9cm}@{}}{\textbf{fractional cost UNI}} &\multicolumn{1}{@{}P{1.8cm}@{}}{\textbf{ratio cost UNI/ours}} \\
\midrule
    10 & 0.0000 & 2.2670E-4 &     $\infty$\\
   100 & 0.0000 & 3.5450E-5 &     $\infty$\\
  1000 & 2.6521E-5 & 5.5047E-3 & 207.5644 \\
 10000 & 1.7543E-4 & 5.0538E-3 &  28.8085 \\
100000 & 3.2008E-4 & 5.0087E-3 &  15.6483 \\
\bottomrule
\end{tabular}
\caption{USDC$\rightleftharpoons$WETH}\label{tab:costsavedUE2}
\end{subtable}
\hfill
\begin{subtable}{0.49\linewidth}
\centering
\begin{tabular}{@{}rn{1}{3}n{1}{3}rr@{}}
\toprule
\textbf{size [\$]} &   \multicolumn{1}{@{}P{2cm}@{}}{\textbf{fractional cost ours}} &  \multicolumn{1}{@{}P{1.9cm}@{}}{\textbf{fractional cost UNI}} &\multicolumn{1}{@{}P{1.8cm}@{}}{\textbf{ratio cost UNI/ours}} \\
\midrule
    10 & 0.0000 & 7.4405E-5 &     $\infty$\\
   100 & 4.5620E-6 & 1.5155E-5 &   3.3219 \\
  1000 & 1.1037E-5 & 5.5056E-3 & 498.8522 \\
 10000 & 4.6593E-5 & 5.0553E-3 & 108.5005 \\
100000 & 8.7459E-5 & 5.0103E-3 &  57.2870 \\

\bottomrule
\end{tabular}
\caption{WBTC$\rightleftharpoons$WETH}\label{tab:costsavedBE2}
\end{subtable}
\vspace{0.2cm}

\begin{subtable}{0.49\linewidth}
\centering
\begin{tabular}{@{}rn{1}{3}n{1}{3}rr@{}}
\toprule
\textbf{size [\$]} &   \multicolumn{1}{@{}P{2cm}@{}}{\textbf{fractional cost ours}} &  \multicolumn{1}{@{}P{1.9cm}@{}}{\textbf{fractional cost UNI}} &\multicolumn{1}{@{}P{1.8cm}@{}}{\textbf{ratio cost UNI/ours}} \\
\midrule
    10 & 0.0000 & 8.3107E-5 &     $\infty$\\
   100 & 0.0000 & 1.3357E-5 &     $\infty$\\
  1000 & 7.7413E-6 & 5.5021E-3 & 710.7457 \\
 10000 & 3.2566E-5 & 5.0518E-3 & 155.1261 \\
100000 & 5.9743E-5 & 5.0067E-3 &  83.8038 \\
\bottomrule
\end{tabular}
\caption{USDC$\rightleftharpoons$USDT}\label{tab:costsavedUU2}
\end{subtable}
\hfill
\begin{subtable}{0.49\linewidth}
\centering
\begin{tabular}{@{}rn{1}{3}n{1}{3}rr@{}}
\toprule
\textbf{size [\$]} &   \multicolumn{1}{@{}P{2cm}@{}}{\textbf{fractional cost ours}} &  \multicolumn{1}{@{}P{1.9cm}@{}}{\textbf{fractional cost UNI}} &\multicolumn{1}{@{}P{1.8cm}@{}}{\textbf{ratio cost UNI/ours}} \\
\midrule
    10 & 0.0000 & 8.0255E-4 &     $\infty$\\
   100 & 5.4652E-7 & 1.3430E-4 & 245.7445 \\
  1000 & 7.5484E-5 & 5.5271E-3 &  73.2224 \\
 10000 & 9.8811E-5 & 5.0736E-3 &  51.3471 \\
100000 & 1.9092E-4 & 5.0276E-3 &  26.3332 \\
\bottomrule
\end{tabular}
\caption{WBTC$\rightleftharpoons$USDC}\label{tab:costsavedBU2}
\end{subtable}
\vspace{0.2cm}

\begin{subtable}{0.49\linewidth}
\centering
\begin{tabular}{@{}rn{1}{3}n{1}{3}rr@{}}
\toprule
\textbf{size [\$]} &   \multicolumn{1}{@{}P{2cm}@{}}{\textbf{fractional cost ours}} &  \multicolumn{1}{@{}P{1.9cm}@{}}{\textbf{fractional cost UNI}} &\multicolumn{1}{@{}P{1.8cm}@{}}{\textbf{ratio cost UNI/ours}} \\
\midrule
    10 & 0.0000 & 3.2072E-4 &     $\infty$\\
   100 & 8.9194E-6 & 7.6060E-5 &   8.5274 \\
  1000 & 5.1423E-5 & 5.5366E-3 & 107.6673 \\
 10000 & 5.0777E-5 & 5.0835E-3 & 100.1143 \\
100000 & 7.0578E-5 & 5.0313E-3 &  71.2869 \\
\bottomrule
\end{tabular}
\caption{UNI$\rightleftharpoons$USDC}\label{tab:costsavedUNU2}
\end{subtable}
\hfill
\begin{subtable}{0.49\linewidth}
\centering
\begin{tabular}{@{}rn{1}{3}n{1}{3}rr@{}}
\toprule
\textbf{size [\$]} &   \multicolumn{1}{@{}P{2cm}@{}}{\textbf{fractional cost ours}} &  \multicolumn{1}{@{}P{1.9cm}@{}}{\textbf{fractional cost UNI}} &\multicolumn{1}{@{}P{1.8cm}@{}}{\textbf{ratio cost UNI/ours}} \\
\midrule
    10 & 0.0000 & 5.7070E-5 &     $\infty$\\
   100 & 1.1519E-5 & 2.0320E-5 &   1.7641 \\
  1000 & 1.6523E-5 & 5.5144E-3 & 333.7329 \\
 10000 & 1.6352E-5 & 5.0642E-3 & 309.7068 \\
100000 & 2.6573E-5 & 5.0191E-3 & 188.8764 \\
\bottomrule
\end{tabular}
\caption{LINK$\rightleftharpoons$WETH}\label{tab:costsavedLE2}
\end{subtable}
\vspace{0.2cm}

\begin{subtable}{0.49\linewidth}
\centering
\begin{tabular}{@{}rn{1}{3}n{1}{3}rr@{}}
\toprule
\textbf{size [\$]} &   \multicolumn{1}{@{}P{2cm}@{}}{\textbf{fractional cost ours}} &  \multicolumn{1}{@{}P{1.9cm}@{}}{\textbf{fractional cost UNI}} &\multicolumn{1}{@{}P{1.8cm}@{}}{\textbf{ratio cost UNI/ours}} \\
\midrule
    10 & 0.0000 & 2.7638E-4 &     $\infty$\\
   100 & 4.5394E-6 & 4.6882E-5 &  10.3277 \\
  1000 & 3.0016E-5 & 5.5100E-3 & 183.5702 \\
 10000 & 8.1380E-5 & 5.0588E-3 &  62.1630 \\
100000 & 1.5620E-4 & 5.0136E-3 &  32.0972 \\
\bottomrule
\end{tabular}
\caption{DPI$\rightleftharpoons$WETH}\label{tab:costsavedDE2}
\end{subtable}
\hfill
\begin{subtable}{0.49\linewidth}
\centering
\begin{tabular}{@{}rn{1}{3}n{1}{3}rr@{}}
\toprule
\textbf{size [\$]} &   \multicolumn{1}{@{}P{2cm}@{}}{\textbf{fractional cost ours}} &  \multicolumn{1}{@{}P{1.9cm}@{}}{\textbf{fractional cost UNI}} &\multicolumn{1}{@{}P{1.8cm}@{}}{\textbf{ratio cost UNI/ours}} \\
\midrule
    10 & 0.0000 & 1.6934E-4 &     $\infty$\\
   100 & 4.5620E-6 & 3.1338E-5 &   6.8693 \\
  1000 & 1.8468E-5 & 5.5091E-3 & 298.3095 \\
 10000 & 2.9145E-5 & 5.0585E-3 & 173.5636 \\
100000 & 4.4243E-5 & 5.0133E-3 & 113.3132 \\
\bottomrule
\end{tabular}
\caption{KIMCHI$\rightleftharpoons$WETH}\label{tab:costsavedKE2}
\end{subtable}\vspace{-6pt}
\caption{Cost comparison when using our own algorithm to set the slippage tolerance vs. the slippage tolerance suggested by Uniswap. The fractional cost includes both the costs of being attacked as well as the costs associated with redoing the transactions. The simulation spans between blocks 11589848 and 1170984. The base fee $b$ is set to \$2.}\label{tab:costsaved2}\vspace{-14pt}

\end{table*}
We repeat the simulation from Section~\ref{sec:cost} with a smaller base fee, i.e., $b=\$2$, and present the results in Table~\ref{tab:costsaved2}. The table shows the fractional cost incurred when trades use our algorithm and the cost incurred by trades using Uniswap's auto-slippage. For the smaller base fee, our algorithm also saves significant amounts of money in comparison to the suggestions from Uniswap. Due to the lower base fee, smaller trades with Uniswap's auto-slippage become attackable. Thus, while Uniswap's auto-slippage appeared reasonable for trades of size $\$1000$ for $b=\$4$, Uniswap's auto-slippage is greatly outperformed our algorithm for trades of size $\$1000$ for $b= \$2$. 

Further, when also considering the simulation results with $b=\$8$ (cf. Table~\ref{tab:costsaved8}), we see that the general pattern stays the same. Independent of the precise base fee, our algorithm outperforms Uniswap's auto-slippage. The auto-slippage is too low for smaller trades and consequently small trades fail unnecessarily. On the other hand, for larger trades Uniswap's auto-slippage is too high and causes all trades to be attackable. It is clear that a constant auto-slippage cannot consistently perform well.

\nprounddigits{3}
\begin{table*}[htbp]
\centering
\begin{subtable}{0.49\linewidth}
\centering
\begin{tabular}{@{}rn{1}{3}n{1}{3}rr@{}}
\toprule
\textbf{size [\$]} &   \multicolumn{1}{@{}P{2cm}@{}}{\textbf{fractional cost ours}} &  \multicolumn{1}{@{}P{1.9cm}@{}}{\textbf{fractional cost UNI}} &\multicolumn{1}{@{}P{1.8cm}@{}}{\textbf{ratio cost UNI/ours}} \\
\midrule
    10 & 0.0000 & 2.2670E-4 &    $\infty$\\
   100 & 0.0000 & 3.5450E-5 &    $\infty$\\
  1000 & 7.7793E-7 & 1.6325E-5 & 20.9851 \\
 10000 & 1.0254E-4 & 5.2034E-3 & 50.7475 \\
100000 & 2.5830E-4 & 5.0236E-3 & 19.4491 \\
\bottomrule
\end{tabular}
\caption{USDC$\rightleftharpoons$WETH}\label{tab:costsavedUE8}
\end{subtable}
\hfill
\begin{subtable}{0.49\linewidth}
\centering
\begin{tabular}{@{}rn{1}{3}n{1}{3}rr@{}}
\toprule
\textbf{size [\$]} &   \multicolumn{1}{@{}P{2cm}@{}}{\textbf{fractional cost ours}} &  \multicolumn{1}{@{}P{1.9cm}@{}}{\textbf{fractional cost UNI}} &\multicolumn{1}{@{}P{1.8cm}@{}}{\textbf{ratio cost UNI/ours}} \\
\midrule
    10 & 0.0000 & 7.4405E-5 &     $\infty$\\
   100 & 0.0000 & 1.5155E-5 &     $\infty$\\
  1000 & 4.5833E-6 & 9.2297E-6 &   2.0138 \\
 10000 & 3.0479E-5 & 5.2052E-3 & 170.7785 \\
100000 & 5.0588E-5 & 5.0253E-3 &  99.3369 \\
\bottomrule
\end{tabular}
\caption{WBTC$\rightleftharpoons$WETH}\label{tab:costsavedBE8}
\end{subtable}
\vspace{0.2cm}

\begin{subtable}{0.49\linewidth}
\centering
\begin{tabular}{@{}rn{1}{3}n{1}{3}rr@{}}
\toprule
\textbf{size [\$]} &   \multicolumn{1}{@{}P{2cm}@{}}{\textbf{fractional cost ours}} &  \multicolumn{1}{@{}P{1.9cm}@{}}{\textbf{fractional cost UNI}} &\multicolumn{1}{@{}P{1.8cm}@{}}{\textbf{ratio cost UNI/ours}} \\
    10 & 0.0000 & 8.3107E-5 &     $\infty$\\
   100 & 0.0000 & 1.3357E-5 &     $\infty$\\
  1000 & 4.8165E-7 & 6.3817E-6 &  13.2498 \\
 10000 & 1.9062E-5 & 5.2017E-3 & 272.8749 \\
100000 & 3.8064E-5 & 5.0217E-3 & 131.9262 \\
\bottomrule
\end{tabular}
\caption{USDC$\rightleftharpoons$USDT}\label{tab:costsavedUU8}
\end{subtable}
\hfill
\begin{subtable}{0.49\linewidth}
\centering
\begin{tabular}{@{}rn{1}{3}n{1}{3}rr@{}}
\toprule
\textbf{size [\$]} &   \multicolumn{1}{@{}P{2cm}@{}}{\textbf{fractional cost ours}} &  \multicolumn{1}{@{}P{1.9cm}@{}}{\textbf{fractional cost UNI}} &\multicolumn{1}{@{}P{1.8cm}@{}}{\textbf{ratio cost UNI/ours}} \\
\midrule
    10 & 0.0000 & 8.0255E-4 &    $\infty$\\
   100 & 0.0000 & 1.3430E-4 &    $\infty$\\
  1000 & 2.5672E-6 & 6.7471E-5 & 26.2815 \\
 10000 & 9.3562E-5 & 5.2225E-3 & 55.8192 \\
100000 & 1.2526E-4 & 5.0425E-3 & 40.2571 \\
\bottomrule
\end{tabular}
\caption{WBTC$\rightleftharpoons$USDC}\label{tab:costsavedBU8}
\end{subtable}
\vspace{0.2cm}

\begin{subtable}{0.49\linewidth}
\centering
\begin{tabular}{@{}rn{1}{3}n{1}{3}rr@{}}
\toprule
\textbf{size [\$]} &   \multicolumn{1}{@{}P{2cm}@{}}{\textbf{fractional cost ours}} &  \multicolumn{1}{@{}P{1.9cm}@{}}{\textbf{fractional cost UNI}} &\multicolumn{1}{@{}P{1.8cm}@{}}{\textbf{ratio cost UNI/ours}} \\
\midrule
    10 & 0.0000 & 3.2072E-4 &     $\infty$\\
   100 & 0.0000 & 7.6060E-5 &     $\infty$\\
  1000 & 3.6318E-5 & 5.1469E-5 &   1.4172 \\
 10000 & 5.0695E-5 & 5.2331E-3 & 103.2274 \\
100000 & 4.6611E-5 & 5.0463E-3 & 108.2632 \\
\bottomrule
\end{tabular}
\caption{UNI$\rightleftharpoons$USDC}\label{tab:costsavedUNU8}
\end{subtable}
\hfill
\begin{subtable}{0.49\linewidth}
\centering
\begin{tabular}{@{}rn{1}{3}n{1}{3}rr@{}}
\toprule
\textbf{size [\$]} &   \multicolumn{1}{@{}P{2cm}@{}}{\textbf{fractional cost ours}} &  \multicolumn{1}{@{}P{1.9cm}@{}}{\textbf{fractional cost UNI}} &\multicolumn{1}{@{}P{1.8cm}@{}}{\textbf{ratio cost UNI/ours}} \\
\midrule
    10 & 0.0000 & 5.7070E-5 &     $\infty$\\
   100 & 0.0000 & 2.0320E-5 &     $\infty$\\
  1000 & 1.6001E-5 & 1.6644E-5 &   1.0402 \\
 10000 & 1.6417E-5 & 5.2141E-3 & 317.6139 \\
100000 & 1.6230E-5 & 5.0341E-3 & 310.1651 \\
\bottomrule
\end{tabular}
\caption{LINK$\rightleftharpoons$WETH}\label{tab:costsavedLE8}
\end{subtable}
\vspace{0.2cm}

\begin{subtable}{0.49\linewidth}
\centering
\begin{tabular}{@{}rn{1}{3}n{1}{3}rr@{}}
\toprule
\textbf{size [\$]} &   \multicolumn{1}{@{}P{2cm}@{}}{\textbf{fractional cost ours}} &  \multicolumn{1}{@{}P{1.9cm}@{}}{\textbf{fractional cost UNI}} &\multicolumn{1}{@{}P{1.8cm}@{}}{\textbf{ratio cost UNI/ours}} \\
\midrule
    10 & 0.0000 & 2.7638E-4 &    $\infty$\\
   100 & 0.0000 & 4.6882E-5 &    $\infty$\\
  1000 & 5.4281E-6 & 2.3932E-5 &  4.4089 \\
 10000 & 5.8527E-5 & 5.2084E-3 & 88.9918 \\
100000 & 1.0210E-4 & 5.0286E-3 & 49.2509 \\
\bottomrule
\end{tabular}
\caption{DPI$\rightleftharpoons$WETH}\label{tab:costsavedDE8}
\end{subtable}
\hfill
\begin{subtable}{0.49\linewidth}
\centering
\begin{tabular}{@{}rn{1}{3}n{1}{3}rr@{}}
\toprule
\textbf{size [\$]} &   \multicolumn{1}{@{}P{2cm}@{}}{\textbf{fractional cost ours}} &  \multicolumn{1}{@{}P{1.9cm}@{}}{\textbf{fractional cost UNI}} &\multicolumn{1}{@{}P{1.8cm}@{}}{\textbf{ratio cost UNI/ours}} \\
\midrule
    10 & 0.0000 & 1.6934E-4 &     $\infty$\\
   100 & 0.0000 & 3.1338E-5 &     $\infty$\\
  1000 & 5.7811E-6 & 1.7537E-5 &   3.0336 \\
 10000 & 2.3531E-5 & 5.2082E-3 & 221.3346 \\
100000 & 3.0411E-5 & 5.0283E-3 & 165.3476 \\
\bottomrule
\end{tabular}
\caption{KIMCHI$\rightleftharpoons$WETH}\label{tab:costsavedKE8}
\end{subtable}\vspace{-6pt}

\caption{Cost comparison when using our own algorithm to set the slippage tolerance vs. the slippage tolerance suggested by Uniswap. The fractional cost includes both the costs of being attacked as well as the costs associated with redoing the transactions. The simulation spans between blocks 11589848 and 1170984. The base fee $b$ is set to \$8.}\label{tab:costsaved8}\vspace{-14pt}

\end{table*}
\end{document}